\documentclass[conference]{IEEEtran}
\IEEEoverridecommandlockouts
\usepackage{cite}
\usepackage{amsmath,amssymb,amsfonts,bm}
\usepackage{graphicx}
\usepackage{textcomp}
\usepackage{xcolor}
\usepackage{graphicx}
\usepackage{amsthm}
\usepackage{algorithm}
\usepackage[noend]{algorithmic}
\usepackage{array}
\usepackage{mdwmath}
\usepackage{mdwtab}
\usepackage{eqparbox}
\usepackage[tight,footnotesize]{subfigure}
\usepackage{stfloats}
\usepackage{tabularx}
\usepackage{multirow}

\usepackage{url}
\usepackage{booktabs}
\usepackage{bbm}

\newtheorem{theorem}{Theorem}
\newtheorem{definition}{Definition}
\newtheorem{lemma}{Lemma}
\newtheorem{corollary}{Corollary}
\newtheorem{proposition}{Proposition}

\newtheorem{observation}{Observation}

\newcommand\blfootnote[1]{%
	\begingroup
	\renewcommand\thefootnote{}\footnote{#1}%
	\addtocounter{footnote}{-1}%
	\endgroup
}

\makeatletter
\newcommand*{\ov}[1]{%
	$\m@th\overline{\raisebox{0pt}[1.3\height]{#1}}$%
}

\allowdisplaybreaks[4]

\begin{document}

%\title{Flow against Flow: Network Interdiction through Traffic Injection}

\title{\hspace{-2.8mm}Network Interdiction Using Adversarial Traffic Flows}
\author{\large Xinzhe Fu and Eytan Modiano\\Lab for Information and Decision Systems, MIT}
\maketitle
\blfootnote{This work was supported by DTRA grants HDTRA1-13-1-0021 and HDTRA1-14-1-0058, and NSF grant CNS-1735463.}
\begin{abstract}
Traditional network interdiction refers to the problem of an interdictor trying to reduce the throughput of network users by removing network edges. 
In this paper, we propose a new paradigm for network interdiction that models scenarios, such as stealth DoS attack, where the interdiction is performed through injecting adversarial traffic flows. Under this paradigm, we first study the deterministic flow interdiction problem, where the interdictor has perfect knowledge of the operation of network users. We show that the problem is highly inapproximable on general networks and is NP-hard even when the network is acyclic. We then propose an algorithm that achieves a logarithmic approximation ratio and quasi-polynomial time complexity for acyclic networks through harnessing the submodularity of the problem. Next, we investigate the robust flow interdiction problem, which adopts the robust optimization framework to capture the case where definitive knowledge of the operation of network users is not available. We design an approximation framework that integrates the aforementioned algorithm, yielding a quasi-polynomial time procedure with poly-logarithmic approximation ratio for the more challenging robust flow interdiction. Finally, we evaluate the performance of the proposed algorithms through simulations, showing that they can be efficiently implemented and yield near-optimal solutions.
\end{abstract}

\section{Introduction}
Network interdiction, originally proposed in \cite{cite:interdiction1,cite:interdiction4} models the scenarios where a budget-constrained \textit{interdictor} tries to limit the throughput available for \textit{users} of a capacitated network by removing network edges. The throughput is given by the optimal value of a single-commodity max-flow problem and the goal of the interdictor is to compute an interdiction strategy that specifies which edges to remove in order to minimize the throughput, or maximize the throughput reduction, subject to the budget constraint. Since the problem is NP-hard even when the network has special topologies, previous works focus on designing approximation algorithms \cite{cite:interdiction4,cite:interdiction6,cite:interdiction7} or formulating integer programs and solving them using traditional optimization techniques (e.g. branch and bound) \cite{cite:interdiction1,cite:interdiction9}. Subsequent generalizations include extensions to the case where the throughput is given by multi-commodity max-flow problem \cite{cite:interdiction3} and allowing the interdictor to use mixed strategy that takes advantage of randomization \cite{cite:interdiction5,cite:interdiction10}. We refer the readers to \cite{cite:interdictionsurvey} for a comprehensive survey.% on network interdiction.

%Updating the reference

%Network interdiction has applications in a wide range of areas that involve the interactions of two conflicting parties in a network context.
%% that involve one party working on reducing the utility of another in a network contexts. 
%Projecting the interdictor into different parties in real life situations, network interdiction provides valuable insights to various aspects of the network operation.

As a generalization of the renowned max-flow min-cut theorem, network interdiction provides valuable insights to the robustness of networks. Projecting the interdictor to an adversarial position, network interdiction can characterize the impact of natural disasters on fiber-optic networks \cite{cite:maxflow}, evaluate the vulnerability of network infrastructures \cite{cite:infrastructure}, and provide guidelines to the design of security-enhancing strategies for cyber-physical systems \cite{cite:cyberphysical}.
%Network interdiction has received considerable interests due to its wide applicability \cite{cite:interdiction1}-\cite{cite:interdiction9}. 
% If we assume the role of the interdictor, it can be used to devise strategy for applications involving inhibiting the flow of adversarial materials such as military planning \cite{cite:interdiction1}, disease infection control \cite{cite:infection} and drug interdiction \cite{cite:interdiction6}.
% If we take the interdictor as an adversary, it presents important measures on the robustness of network infrastructure \cite{cite:infrastructure} and gives guidelines to the design of security-enhancing strategies for cyber-physical systems \cite{cite:cyberphysical}.
%%%更细节的traditional formulation，two players，interdictor，user
%%min/max throughput/reduction

In this paper, we propose a new paradigm for network interdiction where the interdiction is performed through injecting adversarial traffic flow to the network in an intelligent way, encroaching the capacity of network links, thereby reducing the throughput of network users. It captures applications that have eluded the traditional network interdiction paradigm based on edge removals. One of the most prominent examples is the stealth denial of service (DoS) attack in  communication networks, including wireless ad hoc networks \cite{cite:dosAdhoc}, software defined networks \cite{cite:dosSDN} and cloud services \cite{cite:dosCloud}. The interdictor (attacker) injects low-rate data into the network that consumes network resources and compromises the capacity available to the users.
In this paradigm, we model the network as a capacitated directed graph with $n$ nodes, where the network users are sending flow on a set $P$ of user paths and the interdictor aims to reduce the throughput of the users through sending adversarial flow from its source $s$ to its destination $t$. Mirroring the situations in \cite{cite:dosAdhoc,cite:dosSDN,cite:paschos}, we assume the interdictor to be low-rate and undetectable, which will be formally defined later. The interdiction strategy is defined as a probability distribution over the set of $s$-$t$ flows with value less than the given budget, which resembles the mixed strategy in the  game theory literature \cite{cite:interdiction5}. The throughput reduction achieved is equal to the difference between the network throughput before the interdiction, which is defined as the sum of initial flow values on paths in $P$ and the network throughput after the interdiction, which is determined by the optimal value of a path-based max-flow problem on the residual network that subsumes both single and multi-commodity max-flows. 

Under the proposed interdiction paradigm, we study two problems that differ in the availability of the knowledge on the operation of network users captured by the set of user paths $P$. The first, deterministic flow interdiction, assumes that the interdictor has perfect knowledge of $P$ and seeks an interdiction strategy that maximizes the (expected) throughput reduction with respect to $P$. We show that there does not exist any polynomial time algorithm that approximates the problem on general networks within an $O(n^{1-\delta})$ factor for any $\delta>0$ unless P = NP, and the problem is NP-hard even when the network is acyclic. Thus, we focus on designing efficient algorithms with good performance guarantees on acyclic networks. Specifically, utilizing the submodularity of the problem, we propose a recursive algorithm that is capable of achieving $O(\log n)$-approximation. The second problem, robust flow interdiction, deals with the situation where definitive knowledge of $P$ is not available. In particular, we assume that the set of user paths lies in some uncertainty set $\mathcal{U}$ that contains all possible candidates for $P$. The goal of the interdictor is to compute an interdiction strategy that maximizes the throughput reduction for the worst case in $\mathcal{U}$. As a generalization of its deterministic counterpart, robust flow interdiction inherits the computational complexity results and is more challenging to solve due to its inherent maximin objective. In this context, we design an approximation framework that integrates the algorithm for deterministic flow interdiction and yields a quasi-polynomial time procedure with a poly-logarithmic approximation guarantee. Finally, We evaluate the performance of the proposed algorithms through simulations. The simulation results suggest that our algorithms compute solutions that are at least 70\% of the optimal and are efficiently implementable.

The rest of the paper is organized as follows. We formally present our paradigm on acyclic networks in Section \ref{sec:model}. In Sections \ref{sec:deterministic} and \ref{sec:robust}, we introduce formal definitions, show the computational complexity and describe our proposed algorithms for the two flow interdiction problems, respectively. We evaluate the performance of our algorithms through simulations in Section \ref{sec:simulations}. Section \ref{sec:discussion} is devoted to the extension of our paradigm and the interdiction problems to general networks. We conclude the paper in Section \ref{sec:conclusion}.

\section{Network Interdiction Paradigm} \label{sec:model}
In this section, we first formalize our network interdiction paradigm, and then show two important structural properties of it. Note that currently, we focus on acyclic networks, and provide extensions to general networks in Section \ref{sec:discussion}.

Consider a network represented as a directed acyclic graph $G(V,E)$ with vertex set $V$ and edge set $E\subseteq V\times V$. Let $n=|V|$ be the number of nodes and $m=|E|$ be the number of edges. We assume $G$ to be simple (with no multi-edges). Let $C$ be an $|E|$-dimensional non-negative capacity vector with $C(e)$ indicating the capacity of edge $e$. We define $s,t\in V$ as the source and the destination of the interdictor, and assume without loss of generality that they are connected. An $s$-$t$ flow is defined as an $|E|$-dimensional vector $\mathbf{f}$ that satisfies capacity constraints: $\forall e \in E, 0\le \mathbf{f}(e)\le C(e)$ and flow conservation constraints: $\forall v\in V\backslash\{s,t\},\sum_{(u,v)\in E}\mathbf{f}(u,v)=\sum_{(v,u)\in E}\mathbf{f}(v,u)$. We define $val(\mathbf{f})=\sum_{(s,u)\in E}\mathbf{f}(s,u)$ to be the value of $\mathbf{f}$, i.e., the total flow out of the source. 

The interdiction is performed by injecting flow from $s$ to $t$. The interdictor has a flow budget $\gamma$ that specifies the maximum value of flow that it can inject. In this paper, we are primarily concerned with low-rate interdictor, and thus assume that $\gamma\le\min_{e\in E}C(e)$ and is bounded by some polynomial of $n$. Let $\mathcal{F}_{\le\gamma}$ be the set of $s$-$t$ flows $\mathbf{f}$ with $val(\mathbf{f})\le\gamma$. We allow the interdictor to use randomized flow injection, which is captured by the concept of interdiction strategy formally defined below. 
\begin{definition}[Interdiction Strategy]
	An interdiction strategy $w$ is a probability distribution $w:\mathcal{F}_{\le\gamma}\mapsto [0,1]$ such that $\sum_{\mathbf{f}\in\mathcal{F}_{\le\gamma}}w(\mathbf{f})=1$
\end{definition}
The interdiction strategy bears resemblance to the mixed strategy in the game theory literature. It can be alternatively interpreted as injecting flows in a time sharing way. Furthermore, a deterministic flow injection $\mathbf{f}$ (similar to a pure strategy in game theory) can be represented as a strategy with $w(\mathbf{f})=1$. 

Before the interdiction, the network users are sending flow on a set of user paths $P=\{p_1,p_2,\ldots,p_k\}$ in the network. Each path is a subset of edges and we use $e\in p_i$ to represent that edge $e$ is on path $p_i$. The user paths may not share the same source and destination, and are not necessarily disjoint. Initially, the values of the flows on the paths are $\lambda_1,\lambda_2\ldots,\lambda_k$ respectively, which satisfy capacity constraints: $\forall e\in E,\ \sum_{p_i\ni e}\lambda_i \le C(e)$. The network throughput for the users before the interdiction is defined as $\sum_{i=1}^k\lambda_i$. Note that the involvement of the initial flows gives our paradigm the flexibility to capture the case where the users are not fully utilizing the paths before interdiction.

%During the interdiction process, the interdictor determines the interdiction strategy, specifying how to inject flow from $s$ to $t$. The interdictor has a flow budget $\gamma$ that bounds the maximum value of the flow it can inject. In this paper, we are primarily concerned with low-rate interdictor, %\footnote{The relaxation of the low-rate assumption is also discussed in Section \ref{sec:discussion}} 
%and thus assume that $\gamma\le\min_{e\in E}C(e)$ and is bounded by some polynomial of $n$. Let $\mathcal{F}_{\le\gamma}$ be the set of $s$-$t$ flows $\mathbf{f}$ with $val(\mathbf{f})\le\gamma$.
%We define interdiction strategy to be a probability distribution over $\mathcal{F}_{\le \gamma}$, which bears resemblance to the mixed strategy in game theory literature.

%In this paradigm, we assume that the interdictor undetectable and has priority on edge usage, which means that after the interdictor injecting flow $\mathbf{f}$,  the users do not change their paths and 

After the interdictor injects flow $\mathbf{f}$, the residual capacity of the edges becomes $\tilde{C}_{\mathbf{f}}$ such that $\tilde{C}_{\mathbf{f}}(e)=C(e)-\mathbf{f}(e)$ for all $e\in E$. 
%We drop the omit the dependence of $\tilde{C}$ on $\mathbf{f}$ notationally for readability. 
The throughput of the users after interdiction is given by the optimal value of the following (path-based) max-flow problem:
\begin{align}
\quad\text{maximize } & \textstyle\sum_{i}\tilde{\lambda}_i\label{throughput}\\
\text{\textbf{s.t. }} \textstyle\sum_{p_i\ni e}\tilde{\lambda}_i&\le \tilde{C}_{\mathbf{f}}(e),\quad \forall e\in E \label{capacity}\\
& \textstyle0\le \tilde{\lambda}_i\le \lambda_i,\quad \forall i \label{utility}
\end{align}
where constraints (\ref{capacity}) are the capacity constraints after the interdiction and constraints (\ref{utility}) specify that the users will not actively push more flows on the paths after the interdiction, which can be attributed to the undetectability of the interdictor or that the users have no more flow to send.  
Let $T(\mathbf{f},P)$ be the optimal value of (\ref{throughput}). %$T(\mathbf{f},P)$ gives the throughput of the network under the interdicting flow $\mathbf{f}$. 
We define the throughput reduction achieved by injecting flow $\mathbf{f}$ as the difference between the throughput of the network before and after interdiction, i.e., $\Lambda(\mathbf{f},P)=\sum_i\lambda_i-T(\mathbf{f},P)$. Naturally, under an interdiction strategy ${w}$, the expected throughput reduction achieved by the interdictor is defined as $\Lambda({w},P)=\sum_{\mathbf{f}\in\mathcal{F}_{\le\gamma}}w\mathbf{(f)}\Lambda(\mathbf{f},P)$. 

%\begin{itemize}
%	\item Capacitated flow network, DAG(?), low-rate Adversary
%	\item Multiple user-paths, single/multi-user(?)
%	\item Calculation of Residual Capacity and Capacity Reduction
%\end{itemize}
\subsection{Structural Properties of the Paradigm}
The proposed network interdiction paradigm has two important structural properties that will play a key role in the problems we study in subsequent sections.
The first property shows that if we want to maximize the throughput reduction, we can restrict our consideration to the set of $s$-$t$ flows with value $\gamma$. Its proof follows straightforwardly from the monotonicity of throughput reduction with respect to the value of the interdicting flow and that $\gamma\le\min_eC(e)$.
\begin{observation}
	For any $s$-$t$ flow $\mathbf{f}$ with $val(\mathbf{f})<\gamma$, there exists a flow $\mathbf{f}'$ such that $val(\mathbf{f}')=\gamma$ and $\Lambda(\mathbf{f}',P)\ge\Lambda(\mathbf{f},P)$ for all possible $P$.
\end{observation}
%\begin{observation}
%	For both deterministic and robust flow interdiction problem, there exists an optimal interdiction strategy $w$ such that if $w(\mathbf{f})>0$, then $val(\mathbf{f})=\gamma$.
%\end{observation}

We denote $\mathcal{F}_{\gamma}$ to be the set of flows with value $\gamma$.
%By the flow decomposition theorem \cite{cite:networkflow} and that the network is acyclic, each $s$-$t$ flow can be decomposed into the sum of flows that only have positive values on edges of some $s$-$t$ path. 
We further define \textit{single-path flows} as the $s$-$t$ flows that have positive values on edges of one $s$-$t$ path. The second property establishes the optimality of interdiction strategies taking positive value on only single-path flows in the maximization of $\Lambda$.
%In the following, we first show that we can essentially restrict our consideration to the set of single-path flows. Based on this property, we further establish the NP-hardness of both problems.
\begin{proposition}\label{proposition:property}
	For any interdiction strategy $w$, there exists an interdiction strategy $w'$ that is a probability distribution on the set of single-path flows such that $\Lambda(w',P)\ge\Lambda(w,P)$ for all possible $P$.
\end{proposition}

%\begin{proposition}\label{proposition:property1}
%	For the Deterministic Flow Interdiction problem, there exists an optimal (pure) interdiction strategy $w$ such that $w(\mathbf{f}^*)=1$ for some single-path flow $\mathbf{f^*}$. For the Robust Flow Interdiction problem, there exists an optimal strategy $w$ that is a probability distribution over the set of single-path flows.
%\end{proposition}

\begin{proof}
	We prove the proposition through flow decomposition and linear programming duality. See Appendix \ref{appendix:proofproperty} for details.	
\end{proof}

\section{Deterministic Flow Interdiction} \label{sec:deterministic}
In this section, we study the deterministic flow interdiction problem. We first formally define the problem, then prove its computational complexity, and finally introduce our proposed approximation algorithm.

\subsection{Problem Formulation}

%As the problem is NP-hard, it is infeasible to hope for efficient exact algorithms. Therefore, we aim at proposing approximation algorithm that has both good performance guarantee and relatively low time complexity.

%Based on the network interdiction paradigm, we consider two problems that differ in the level of the interdictor's knowledge on the set of user paths $P$.

The deterministic flow interdiction deals with the case where the interdictor has full knowledge of $P$ and seeks the interdiction strategy that causes the maximum expected throughput reduction.

\begin{definition}[Deterministic Flow Interdiction] \label{def:deterministic}
	Given the set of user paths $P=\{p_1,\ldots,p_k\}$ with initial flow values $\{\lambda_1,\ldots,\lambda_k\}$, the deterministic flow interdiction problem seeks an interdiction strategy $w$ that maximizes $\Lambda(w,P)$.
\end{definition}

%\textbf{Remark:} As $(w,P)=\sum w(\mathbf{f})(\sum_if_i+(-T(\mathbf{f},P)))$, and $(-T(\mathbf{f},P))$ can be seen as the optimal value of a minimization problem, the two interdiction problems can both be written as a maximin optimization problem with linear constraint. By the standard technique that dualizes the inner minimization problem and linearizes the objective function by bringing in integral variables, we can formulate them as integer linear programming (ILP) problems. But neither do the integer programs offer much insight, nor do they play a role in our algorithms, so we will not pursue that direction.

\begin{figure}
	\centering
	\includegraphics[width=0.75\linewidth]{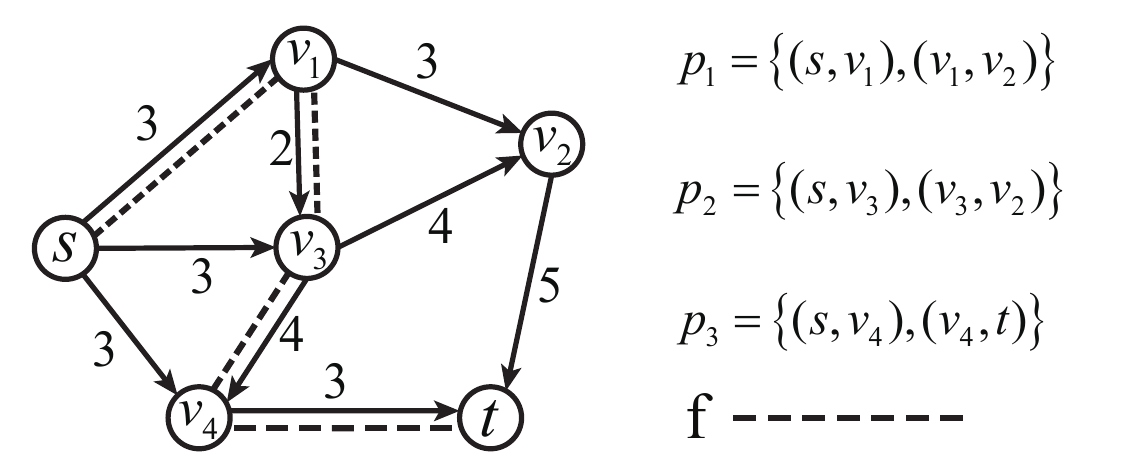}
	\vspace{-3mm}
	\caption{An example network of the flow interdiction problems.}
	\vspace{-2mm}
	\label{fig:example}
\end{figure}

\textbf{Example: }We give an example of the problem. Consider the network in Figure \ref{fig:example}, where the capacities are labeled along the edges. The source and the destination of the interdictor are nodes $s$ and $t$. The interdictor has budget $\gamma=2$. The user paths $P=\{p_1, p_2, p_3\}$ all have an initial flow value of 3. Let $\mathbf{f}$ be the $s$-$t$ flow such that $\mathbf{f}(s,v_1)=\mathbf{f}(v_1,v_3)=\mathbf{f}(v_3,v_4)=\mathbf{f}(v_4,t)=2$ . In this example, the interdiciton strategy $w$ such that $w(\mathbf{f})=1$ is optimal with $\Lambda(w,P)=4$.

\subsection{Computational Complexity}

Before establishing the computational complexity, we first show some structural properties specific to the deterministic flow interdiction problem. Following from Proposition \ref{proposition:property}, there exists an interdiction strategy on the set of single-path flows that is optimal for the deterministic flow interdiction. We further extend this property, showing that there exists an optimal pure interdiction strategy.

\begin{proposition}\label{proposition:property1}
	For the deterministic flow interdiction problem, there exists an optimal (pure) interdiction strategy $w$ such that $w(\mathbf{f}^*)=1$ for some single-path flow $\mathbf{f^*}$.
\end{proposition}
\begin{proof}
	Building on proposition \ref{proposition:property}, let $w$ be an optimal interdiction strategy that takes positive values only on single-path flows $\mathbf{q}_1,\ldots,\mathbf{q}_r$. Let $\mathbf{q}^*\in \arg\max_i\Lambda(\mathbf{q}_i,P)$. Consider the pure strategy $w'$ with $w'(\mathbf{q}^*)=1$. It follows that $\Lambda(w',P)=\Lambda(\mathbf{q}^*,P)\ge\sum_iw(\mathbf{q}_i)\Lambda(\mathbf{q}_i,P)=\Lambda(w,P)$, which proves the existence of an optimal pure strategy.
\end{proof}

From the proof of Propositions \ref{proposition:property} and \ref{proposition:property1}, we can obtain the following corollary.

%%%In plain English!!!
\begin{corollary}\label{corollary:hardness}
	Given an optimal strategy $w$ to the deterministic flow interdiction problem, we can obtain another optimal strategy $w'$ with $w'(\mathbf{f}^*)=1$ for some single-path flow $\mathbf{f^*}$. 
\end{corollary}
\begin{proof}
	For any $\mathbf{f}$ that $w(\mathbf{f})>0$, it can be decomposed into single-path flows $\mathbf{q}_1,\ldots,\mathbf{q}_r$. From the proofs of  Propositions \ref{proposition:property} and \ref{proposition:property1}, it follows that strategies $w_i, i\in\{1,\ldots,r\}$ with $w_i(\frac{\gamma}{val(\mathbf{q}_i)}\mathbf{q}_i)=1$ are all optimal.
\end{proof}

Corollary \ref{corollary:hardness} states that a single path flow that maximizes $\Lambda(\cdot,P)$ can be obtained from an optimal interdiction strategy for the deterministic flow interdiciton problem in polynomial time. Hence, the NP-hardness of finding an optimal single-path flow implies the NP-hardness of the deterministic flow interdiction. Based on this result, we prove the NP-hardness of the deterministic flow interdiction problems.
\begin{proposition}\label{proposition:hardness}
	The deterministic flow interdiction problem is NP-hard.
\end{proposition}
\begin{proof}
	 The proof is done by reduction from the 3-satisfiability problem, which is a classical NP-Complete problem \cite{cite:karp}. See Appendix \ref{appendix:proofhardness} for the details.
\end{proof}
\textbf{Remark:} From the proof of Proposition \ref{proposition:hardness}, we have that even when the user paths are disjoint, the deterministic problem is still NP-hard.
%\textbf{Remark:} Two remarks are due at this point.
%\begin{enumerate}
% \item From the proof of Proposition \ref{proposition:hardness}, we have that even when the user paths are disjoint, the deterministic problem is still NP-hard.
% \item As we mentioned, on general directed graphs, the deterministic flow interdiction problem is highly inapproximable. We will show in Section \ref{sec:discussion}, that there does not exist any polynomial time (quasi-polynomial time) algorithm with an approximation ratio $O(n^{1-\delta})$ for any $\delta>0$ unless $P=NP$ $(DTIME(n^{\log n})=NP)$\footnote{$DTIME(n^{\log n})$ denotes the class of problems that can be solved in quasi-polynomial time.}. This implies that it is NP-hard to approximate the problem within a non-trivial factor under general networks. Therefore, we keep focusing on designing efficient approximation algorithm with good performance guarantee for the problem on acyclic networks.
% 
%\end{enumerate}

%By Proposition 1, there exists a single-path flow that is optimal for the Deterministic Flow Interdiction.
%The algorithm we propose exploits this structural results in Section \ref{sec:structural} and lays the foundation for the algorithm for the Robust Flow Interdiction. 

\subsection{Approximation Algorithm}
Before presenting the algorithm, we extend some previous definitions. For any subset of edges $A\subseteq E$, imagine that the interdictor can interdict the edges in $A$ by reducing their capacities by $\gamma$. We extend the definition of $\Lambda(\cdot,P)$ to $A$ as $\Lambda(A,P)=\sum_i{\lambda}_i-T(A,P)$, where $T(A,P)$ is the optimal value of the maximization problem (\ref{throughput}) with $\tilde{C}_A(e)=C(e)-\gamma\cdot\mathbbm{1}_{\{e\in A\}}$. This provides an interpretation of $\Lambda(\cdot,P)$ as a set function on all subsets of $E$.
 Note that each single-path flow $\mathbf{f}$ can be equivalently represented as a set of edges $E_{\mathbf{f}}$ with $\mathbf{f}(e)=\gamma$ if and only if $e\in E_{\mathbf{f}}$. It follows that $\Lambda(\mathbf{f},P)=\Lambda(E_{\mathbf{f}},P)$, which links the definition of $\Lambda(\cdot,P)$ on single-path flows to that on sets of edges. 
 %%Interchangeably
 
 Our algorithm works on the optimization problem below.
\begin{align}
\text{maximize} \quad&\Lambda(E_\mathbf{f},P)\label{transform}\\
\textbf{s.t.}\quad & E_\mathbf{f} \mbox{ forms an $s$-$t$ path.} \nonumber
\end{align}
Let $E_\mathbf{f^*}$ be the optimal solution to (\ref{transform}) and $\mathbf{f}^*$ be its corresponding single-path flow. By Proposition \ref{proposition:property}, the strategy $w$ with $w(\mathbf{f^*})$ is an optimal interdiction strategy, and $\Lambda(E_{\mathbf{f}^*},P)=\Lambda(w,P)$. Therefore, through approximating problem (\ref{transform}), our algorithm translates to an approximation to the deterministic flow interdiction problem. In the sequel, to better present the main idea of our algorithm, we first discuss the case where the user paths are edge-disjoint. After that, we generalize the results to the non-disjoint case.

%\subsubsection{Disjoint User Paths}
\textit{C.1) Disjoint User Paths:}
When the user paths are edge-disjoint, for some interdicted edges $A\subseteq E$ and user paths $\{p_1,\ldots,p_k\}$ with initial values $\{\lambda_1,\ldots,\lambda_k\}$, the optimal solution to the max-flow problem (\ref{throughput}) can be easily obtained as $\tilde{\lambda}_{iA}=\min\left(\lambda_i,\min_{e\in p_i}\tilde{C}_A(e)\right)$ for all $i$. It follows that the throughput reduction can be written as the sum of the throughput reduction on each paths, i.e., $\Lambda(A,P)=\sum_i(\lambda_i-\tilde{\lambda}_{iA})$. Based on this, we reason below that the set function $\Lambda(\cdot,P)$ has two important properties: \textit{monotonicity and submodularity}.

\begin{lemma}\label{lemma:submodular}
	Consider $\Lambda(\cdot,P):2^E\mapsto \mathbf{R}^*$ as a set function. $\Lambda$ is:
	\begin{enumerate}
		\item Monotone: $\Lambda(A,P)\le \Lambda(B,P)$ for all $A\subseteq B$;
		\item Submodular: for all $A,B\subseteq E, e\in E$, if $A\subseteq B$, then $\Lambda(A\cup\{e\},P)-\Lambda(A,P)\ge \Lambda(B\cup\{e\},P)-\Lambda(B,P)$.
	\end{enumerate}
\end{lemma}
\begin{proof}
	The monotonicity is easily seen from the definition of $\Lambda$. The proof of submodularity is also straightforward. Note that for each user path $i$,
	\begin{align*}
	 \lambda_i-\tilde{\lambda}_{iA}&=\lambda_i-\min(\lambda_i,\min_{e\in p_i}\{C(e)-\mathbbm{1}_{\{e\in A\}}\})\\
	 &=\lambda_i+\max(-\lambda_i,\max_{e\in p_i}\{-C(e)+\mathbbm{1}_{\{e\in A\}}\})
	 \end{align*}
	  Since constant and linear functions are submodular, and the maximum of a set of submodular functions is also submodular, it follows that $\Lambda(A,P)=\sum_i(\lambda_i-\tilde{\lambda}_{iA})$ is submodular.
\end{proof}

Intuitively, an $s$-$t$ path with large throughput reduction should have many intersections with different user paths. This intuition, combined with the monotonicity and submodularity of $\Lambda$, may suggest an efficient greedy approach to the optimization problem (\ref{transform}) that iteratively selects the edge with the maximum marginal gain with respect to $\Lambda$ while sharing some $s$-$t$ path with the edges that have already been selected. However, this is not the whole picture since such greedy selection might get stuck in some short $s$-$t$ path and lose the chance of further including the edges that contribute to the throughput reduction. The latter aspect indicates the necessity of extensive search over the set of all $s$-$t$ paths, but the number of $s$-$t$ paths grows exponentially with $n$. Therefore, an algorithm with good performance guarantee and low time complexity must strike a balance between greedy optimization that harnesses the properties of $\Lambda$, and extensive search that avoids prematurely committing to some short path.
The algorithm we propose, named as the Recursive Greedy algorithm, achieves such balance. It is based on the idea of \cite{cite:submodular}. The details of the algorithm are presented in \textbf{Algorithm \ref{algorithm:greedy}} . In the description and analysis of the algorithm, $\Lambda_X(A,P)=\Lambda(A\cup X,P)-\Lambda(X,P)$ for $X,A\subseteq E$ represents the marginal gain of set $A$ with respect to $X$. We use $\log$ to denote the logarithm with base two. For two nodes $u_1,u_2\in V$, the shortest $u_1$-$u_2$ path is defined as the $u_1$-$u_2$ path with the smallest number of edges. 

%%More regirous definition, side marginal cases.

%%%Improve this part, more clear description!

The recursive function $RG$ lies at the heart of the Recursive Greedy algorithm. $RG$ takes four parameters: source $u_1$, destination $u_2$, constructed subpath $X$ and recursion depth $i$. It constructs a path from $u_1$ to $u_2$ that has a large value of $\Lambda_X(\cdot,P)$ by recursively searching for a sequence of good anchors and greedily concatenating the sub-paths between anchors. The base case of the recursion is when the depth $i$ reaches zero, then $RG$ returns the shortest path between $u_1$ and $u_2$ if there exists one (step 2). Otherwise, it goes over all the nodes $v$ in $V$ (step 8), using $v$ as an anchor to divide the search into two parts. For each $v$, it first calls a sub-procedure to search for sub-path from $u_1$ to $v$ that maximizes $\Lambda_X(\cdot,P)$, with $i$ decremented by 1 (step 9). After the first sub-procedure returns $E_{\mathbf{f}_1}$, it calls a second sub-procedure for sub-paths from $v$ to $u_2$ (step 10). Note that the second sub-procedure is performed on the basis of the result of the first one, which reflects the greedy aspect of the algorithm. The two sub-paths concatenated serve as the $u_1$-$u_2$ path that $RG$ obtains for anchor $v$. Finally, $RG$ returns the path that maximizes $\Lambda_X(\cdot,P)$ over the ones that it has examined over all anchors (steps 11, 12 and 13).

The Recursive Greedy algorithm starts by invoking $RG(s,t,\emptyset,I)$ with $I$ as the initial recursion depth. In the following, we show that the algorithm achieves a desirable performance guarantee as long as $I$ is greater than certain threshold.
An illustration of the algorithm with $I=2$ on the previous example is shown in Figure \ref{fig:recursion1}. The optimal solution is returned by the path with anchors $v_1,v_3,v_4$.

\begin{figure}
	\centering
	\includegraphics[width=1\linewidth]{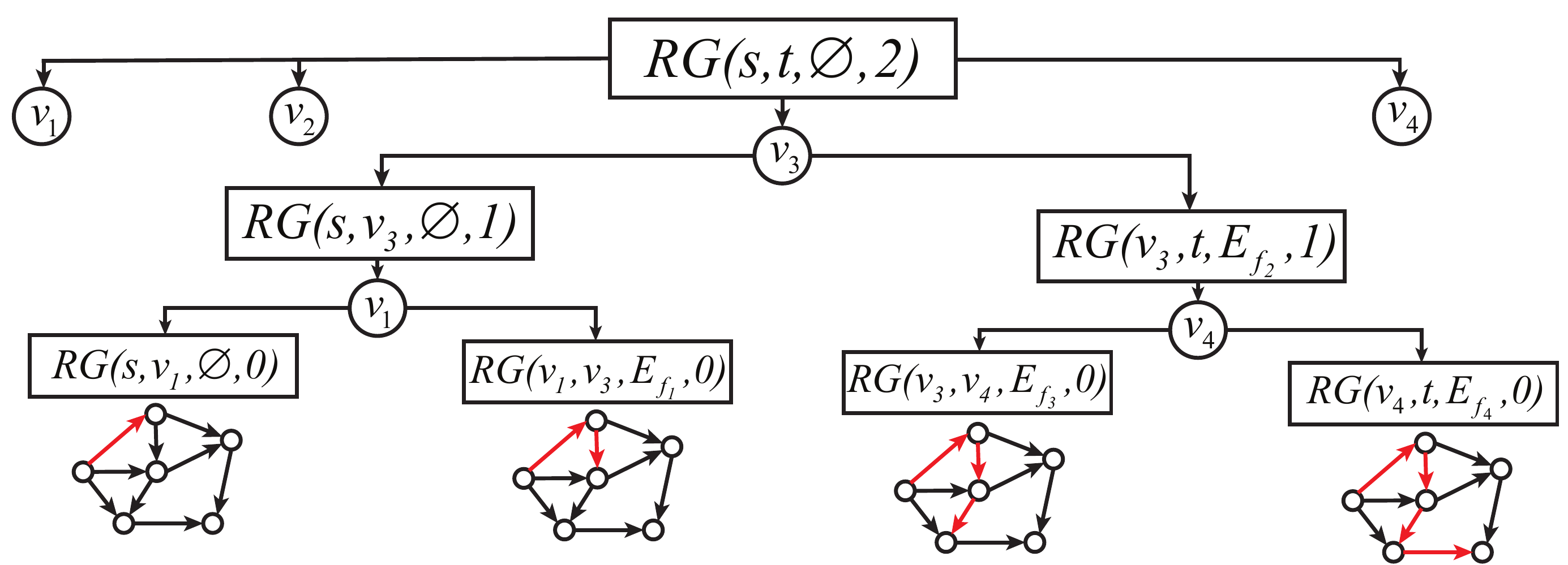}
	\vspace{-3mm}
	\caption{Illustration of the Recursive Greedy algorithm (with some intermediate steps omitted) on the example of Figure \ref{fig:example}, where $E_{\mathbf{f}_1},\ldots,E_{\mathbf{f}_4}$ are used to denote the sub-paths constructed during the recursion for ease of notation. }
	\vspace{-2mm}
	\label{fig:recursion1}
\end{figure}

\begin{algorithm}
	\caption{The Recursive Greedy Algorithm}
	\begin{algorithmic}[1]	\label{algorithm:greedy}		
		\REQUIRE{Network graph $G(V,E)$,  user paths $P=\{p_1,\ldots,p_2\}$ with initial flow values $\{f_1,\ldots,f_k\}$, Interdictor's source $s$, destination $t$ and budget $\gamma$ }\\
		\ENSURE{The optimal $s$-$t$ path $E_\mathbf{f}$}
		
		%			\Procedure {DesignTranScheme}{$\mathcal{T}$}
		%			\EndProcedure
		%Precompute the elements in the corresponding Finite Horizon Markov Decision Process.\\
		\STATE \textbf{Run: } $RG(s,t,\emptyset,I)$\\
		\textit{The Recursive Function $RG(u_1,u_2,X,i)$:}
		\STATE $E_\mathbf{f}:=$ shortest $u_1$-$u_2$ path.
		\IF{$E_\mathbf{f}$ does not exist}
		\RETURN Infeasible
		\ENDIF
		\IF{$i=0$}
		\RETURN $E_\mathbf{f}$
		\ENDIF
		\STATE $r:=\Lambda_X(E_\mathbf{f},P)$.
		\FOR{all $v\in V$}
		\STATE $E_{\mathbf{f}_1}:=RG(u_1,v,X,i-1)$.
		\STATE $E_{\mathbf{f}_2}:=RG(v,u_2,X\cup E_{\mathbf{f}_1},i-1)$.
		\IF{$\Lambda_X(E_{\mathbf{f_1}}\cup E_{\mathbf{f_2}},P)>r$}
		\STATE $r:=\Lambda_X(E_{\mathbf{f_1}}\cup E_{\mathbf{f_2}},P)$, $E_\mathbf{f}:=E_{\mathbf{f}_1}\cup E_{\mathbf{f}_2}$.
		\ENDIF		
		\ENDFOR
		\RETURN{$E_\mathbf{f}$}
	\end{algorithmic}
\end{algorithm}

%%Invoke RG(s,t,\emptyset,I)

%%Statement of the theorem too cluttered or not?
\begin{theorem}\label{thm:recgreedy}
	If $I\ge \lceil\log d\rceil$, the Recursive Greedy algorithm returns an $s$-$t$ path $E_\mathbf{f}$ with $\Lambda(E_\mathbf{f},P)\ge \frac{1}{\lceil\log d\rceil+1}\Lambda(E_\mathbf{f^*},P)$, where $d$ is the length of $E_{\mathbf{f}^*}$.
\end{theorem}

%%%Introduction, and this proof!
\begin{proof}
%	First, we show that the edge set $E_{\mathbf{f}}$ forms an $s$-$t$ path. Indeed, from the description of the recursive function, $E_{\mathbf{f}}$ connects $s$ and $t$ and does not visit any node repeatedly since the network is acyclic. 
	
	We prove a more general claim, that for all $u_1,u_2\in V,\ X\subseteq E$,  if $I\ge \lceil\log d\rceil$, the procedure $RG(u_1,u_2,X,I)$ returns an $u_1$-$u_2$ path $E_\mathbf{f}$ with $\Lambda_X(E_\mathbf{f},P)\ge \frac{1}{\lceil\log d\rceil+1}\Lambda_X(E_\mathbf{f^*},P)$, where $E_{\mathbf{f}^*}$ is the $u_1$-$u_2$ path that maximizes $\Lambda(\cdot,P)$ and $d$ is the length of $E_{\mathbf{f}^*}$. The theorem follows from the claim by setting $u_1=s$, $u_2=t$ and $X=\emptyset$.
	
	Let the nodes on the path $E_\mathbf{f^*}$ be $\{u_1=v_0,\ldots,v_d=u_2\}$. The proof is done by induction on $d$.	
	First, for the base step, when $d=1$, it means that there exists an edge between $u_1$ and $u_2$, which must be the shortest $u_1$-$u_2$ path. Obviously the procedure examines this path at step 2, and the claim follows. Next, suppose the claim holds for $d\le l$. When $d=l+1$, $I\ge 1$. Let $v^*=v_{\lceil \frac{d}{2}\rceil}$ and $E_{\mathbf{f}_1^*},E_{\mathbf{f}_2^*}$ be the subpaths of $E_{\mathbf{f}^*}$ from $u_1$ to $v^*$ and $v^*$ to $t$, respectively. When $RG$ uses $v^*$ as an anchor, it first invokes $RG(u_1,v^*,X,I-1)$ that returns $E_{\mathbf{f}_1}$ and then invokes $RG(v^*,u_2,X\cup E_{\mathbf{f}_1},I-1)$ that returns $E_{\mathbf{f}_2}$. Let $E'_\mathbf{f}=E_{\mathbf{f}_1}\cup E_{\mathbf{f}_2}$. Our goal is to show that 
	\begin{align}
	\Lambda_X(E'_\mathbf{f},P)\ge\frac{1}{\lceil\log d\rceil+1}\Lambda_X\left(E_{\mathbf{f}^*},P\right), \label{ieq:target}
	\end{align}
	which proves the induction step, since the path $E_{\mathbf{f}}$ that $RG(u_1,u_2,X,I)$ returns must satisfy $\Lambda_X(E_{\mathbf{f}},P)\ge \Lambda_X(E'_{\mathbf{f}},P)$.
	
	Since $I\ge \lceil\log d\rceil$, we have $I-1\ge\lceil\log d\rceil-1 = \lceil\log \frac{d}{2}\rceil= \lceil\log \lceil\frac{d}{2}\rceil\rceil$. As $E_{\mathbf{f}_1^*}$ is a path of length $\lceil d/2\rceil$ from $u_1$ to $v^*$ and $E_{\mathbf{f}_2^*}$ is a path of length $\lfloor d/2\rfloor$ from $v^*$ to $u_2$, by the induction hypothesis,
	\begin{align*}
	 \Lambda_{X}(E_{\mathbf{f}_1},P)&\ge \frac{1}{\lceil\log d\rceil}\Lambda_{X}(E_\mathbf{f_1^*},P),\\
	 \Lambda_{X\cup E_{\mathbf{f}_1}}(E_{\mathbf{f}_2},P)&\ge \frac{1}{\lceil\log d\rceil}\Lambda_{X\cup E_{\mathbf{f}_1}}(E_\mathbf{f_2^*},P).
	 \end{align*}
	By the submodularity of $\Lambda$ (Lemma \ref{lemma:submodular}), we have
	\begin{align*}
	 \Lambda_{X}(E_\mathbf{f_1^*},P)\ge \Lambda_{X\cup E'_{\mathbf{f}}}(E_\mathbf{f_1^*},P)\\ 
	  \Lambda_{X\cup E_{\mathbf{f}_1}}(E_\mathbf{f_2^*},P)\ge \Lambda_{X\cup E'_{\mathbf{f}}}(E_\mathbf{f_2^*},P)
	 \end{align*}
	 Using this, we sum the two inequalities obtained form the induction hypothesis and get
%	\begin{align*}
%	&\Lambda_X(E'_\mathbf{f},P)\ge \Lambda_X(E_{\mathbf{f}_1},P) + \Lambda_{X\cup E_{\mathbf{f}_1}}(E_{\mathbf{f}_2},P)\\
%	&\ge \frac{1}{\lceil\log d\rceil}\left( \Lambda_X(E_{\mathbf{f}^*_1},P) +  \Lambda_{X\cup E_{\mathbf{f}_1}}(E_{\mathbf{f}^*_2},P)\right)\\
%	&\ge \frac{1}{\lceil\log d\rceil}\left(\Lambda_{X\cup E'_{\mathbf{f}}}(E_{\mathbf{f}^*_1},P) +  \Lambda_{X\cup E'_{\mathbf{f}}}(E_{\mathbf{f}^*_2},P)\right).
%	\end{align*}
	\begin{small}
	\begin{align*}
	&\Lambda_X(E'_\mathbf{f},P)\ge \frac{1}{\lceil\log d\rceil}\left( \Lambda_X(E_{\mathbf{f}^*_1},P) +  \Lambda_{X\cup E_{\mathbf{f}_1}}(E_{\mathbf{f}^*_2},P)\right)\\
	&\ge \frac{1}{\lceil\log d\rceil}\left(\Lambda_{X\cup E'_{\mathbf{f}}}(E_{\mathbf{f}^*_1},P) +  \Lambda_{X\cup E'_{\mathbf{f}}}(E_{\mathbf{f}^*_2},P)\right).
	\end{align*}
	\end{small}
	%%Here!
	Again, by Lemma \ref{lemma:submodular}, we have,
	\[
	\Lambda_{X\cup E'_{\mathbf{f}}}(E_{\mathbf{f}^*_2},P) \ge \Lambda_{X\cup E'_{\mathbf{f}}\cup E_{\mathbf{f}_1^*}}(E_{\mathbf{f}^*_2},P)
	\]
	It follows that 
%	\begin{small}	
%	\begin{align}
%	 &\Lambda_X(E'_\mathbf{f},P)\ge  \frac{1}{\lceil\log d\rceil}\left( \Lambda_{X\cup E'_{\mathbf{f}}}(E_{\mathbf{f}^*_1},P) +  \Lambda_{X\cup E_{\mathbf{f}_1}\cup E_{\mathbf{f}_2}}(E_{\mathbf{f}^*_2},P)\right)\nonumber\\
%	 &\ge \frac{1}{\lceil\log d\rceil}\left( \Lambda_{X\cup E'_{\mathbf{f}}}(E_{\mathbf{f}^*_1},P) +  \Lambda_{X\cup E'_{\mathbf{f}}\cup E_{\mathbf{f}_1^*}}(E_{\mathbf{f}^*_2},P)\right)\nonumber\\
%	&=\frac{1}{\lceil\log d\rceil}\left(\Lambda\left(X\cup E'_{\mathbf{f}}\cup E_{\mathbf{f}^*},P\right)-\Lambda\left(X\cup E'_{\mathbf{f}},P\right)\right)\label{eq:def}\\
%	&\ge\frac{1}{\lceil\log d\rceil} \left(\Lambda\left(X\cup E_{\mathbf{f}^*},P\right)-\Lambda\left(X\cup E'_{\mathbf{f}},P\right)\right)\label{ieq:monotonicity}\\
%	&= \frac{1}{\lceil\log d\rceil}\left(\Lambda_X\left(E_{\mathbf{f^*}},P\right)-\Lambda_X\left(E'_{\mathbf{f}},P\right)\right)\label{eq:finalbound}, 
%	\end{align}
%	\end{small}
	\begin{small}	
	\begin{align}
	&\Lambda_X(E'_\mathbf{f},P)\ge \frac{1}{\lceil\log d\rceil}\left( \Lambda_{X\cup E'_{\mathbf{f}}}(E_{\mathbf{f}^*_1},P) +  \Lambda_{X\cup E'_{\mathbf{f}}\cup E_{\mathbf{f}_1^*}}(E_{\mathbf{f}^*_2},P)\right)\nonumber\\
	&=\frac{1}{\lceil\log d\rceil}\left(\Lambda\left(X\cup E'_{\mathbf{f}}\cup E_{\mathbf{f}^*},P\right)-\Lambda\left(X\cup E'_{\mathbf{f}},P\right)\right)\label{eq:def}\\
	&\ge\frac{1}{\lceil\log d\rceil} \left(\Lambda\left(X\cup E_{\mathbf{f}^*},P\right)-\Lambda\left(X\cup E'_{\mathbf{f}},P\right)\right)\label{ieq:monotonicity}\\
	&= \frac{1}{\lceil\log d\rceil}\left(\Lambda_X\left(E_{\mathbf{f^*}},P\right)-\Lambda_X\left(E'_{\mathbf{f}},P\right)\right)\label{eq:finalbound}, 
	\end{align}
\end{small}
	where equality (\ref{eq:def}) follows from the definition of $\Lambda_X$ and that $E_{\mathbf{f}^*}=E_{\mathbf{f}^*_1}\cup E_{\mathbf{f}^*_2}$, inequality (\ref{ieq:monotonicity}) follows from the monotonicity of $\Lambda$ and equality (\ref{eq:finalbound}) follows also from the definition $\Lambda_X$.
	From (\ref{eq:finalbound}), we obtain (\ref{ieq:target}), which concludes the proof.
\end{proof}

%%Better analysis of the time complexity
\textbf{Time Complexity: }The bound on the Recursive Greedy algorithm's running time is easy to establish. As we invoke at most $2n$ sub-procedures at each level of recursion and the computation of $\Lambda$ takes $O(m)$ time, the time complexity of the algorithm is $O((2n)^{I}m)$. Taking $I=\log n \ge \lceil\log d\rceil$,\footnote{Strictly speaking, we need to set $I=\lceil \log n\rceil$. We omit the ceiling function here for ease of notations.} we get an algorithm with a logarithmic approximation ratio of $1/(\lceil\log d\rceil+1)$ with a quasi-polynomial time complexity of $O((2n)^{\log n}m)$.

%As seen from Theorem \ref{thm:recgreedy}, parameter $a$ controls the tradeoff between the approximation ratio and the time complexity of the algorithm. A simple instantiation of the algorithm is to set $a=2$ and get an algorithm with a logarithmic approximation ratio of $1/(\lceil\log d\rceil+1)$ with a quasi-polynomial time complexity of $O(n^{(2\log n)}m)$.

\textbf{Remark:} First, note that the proof of Theorem \ref{thm:recgreedy} only relies on the monotonicity and submodularity of $\Lambda$. Therefore, the Recursive Greedy algorithm works for any monotone and submodular function on the subsets of $E$. Second, we can generalize \textbf{Algorithm \ref{algorithm:greedy}} to one that uses more than one anchors at step 8. The generalization is given in Appendix \ref{app:generalization}. When the algorithm uses $a-1$ anchors, it achieves an approximation ratio of $1/(\lceil \log_ad\rceil+1)$ in $O((an)^{(a-1)\log_a n}m)$ time. The parameter $a$ can thus control the tradeoff between the performance guarantee and the time complexity of the algorithm.

%\subsubsection{Non-disjoint User Paths}
\textit{C.2) Non-disjoint User Paths: }When the user paths are not disjoint, the problem becomes more challenging. First, notice that $\tilde{\lambda}_{iA}=\min\left(\lambda_i,\min_{e\in p_i}\tilde{C}_A(e)\right)$ no longer holds due to the constraints in (\ref{capacity}) that couple different $\tilde{\lambda}_i$'s together. More importantly, $\Lambda$ actually loses the submodular property when the user paths are not disjoint, which prevents the direct application of the Recursive Greedy algorithm. We tackle the issues through approximating $\Lambda$ with a monotone and submodular function $\bar{\Lambda}$, and run the Recursive Greedy algorithm on $\bar{\Lambda}$. The performance guarantee of the algorithm can be obtained by bounding the gap between $\Lambda$ and $\bar{\Lambda}$.

Let $E_0\subseteq E$ be the set of edges that belong to some user path. This is also the set of edges that appear in constraints (\ref{capacity}). We partition $E_0$ into two sets $E_1$ and $E_2$, where $E_1$ is the set of edges that belong to only one user path, and $E_2$ is the set of edges that belong to at least two (intersecting) user paths. Following this, we define $\bar{\Lambda}(A,P),A\subseteq E$ to be evaluated through the two-phase procedure below. The procedure first goes edges in $E_1$ (Phase I), setting
\[
\tilde{\lambda}_{iA}^{(1)}:=\min\left(\lambda_i,\min_{e\in p_i,e\in E_1}\{\tilde{C}_A(e)\}\right),\quad \forall i.
\]
 Then, it goes over edges in $E_2$ (Phase II), setting
 \[
 \tilde{\lambda}_{iA}^{(2)}:=\tilde{\lambda}_{iA}^{(1)}\cdot\prod_{e\in p_i,e\in E_2,\tilde{C}_A(e)\le \sum_{p_j\ni e}\lambda_j}\frac{\tilde{C}_A(e)}{\sum_{p_j\ni e}\lambda_j},\quad\forall i.
 \]
 Finally, it sets $\bar{\Lambda}(A,P)=\sum_i\lambda_i-\sum_i\tilde{\lambda}_{iA}^{(2)}$. 
 
 The procedure uses $\{\tilde{\lambda}_{iA}^{(2)}\}$, a set of flow values on user paths, as an approximate solution to the max-flow problem (\ref{throughput}). The solution is obtained through first setting the flow values to $\{\lambda_i\}$ and then gradually decreasing them until the constraints are satisfied. In Phase I, the flow values are decreased to satisfy the capacity constraints posed by edge in $E_1$. In Phase II, the flow values are further reduced to compensate for the capacity violations on edges in $E_2$ through multiplying a factor $\frac{\tilde{C}_A(e)}{\sum_{p_j\ni e}\lambda_j}$, which is equal to the ratio between the capacity of $e$ after the interdiction and the sum of flow values on $e$ before the interdiction, to the flow value of each user path containing $e$, for each $e\in E_2$. Typically, Phase II overcompensates and thus $\bar{\Lambda}$ is an upper bound of $\Lambda$. But as we will show, the gap between $\bar{\Lambda}$ and $\Lambda$ is moderate and such overcompensation guarantees the submodularity of $\bar{\Lambda}$.
%\begin{itemize}
%	\item \textbf{Phase I:} For $i\in\{1,\ldots,k\}: $
%	\[
%	\tilde{f}_{iA}^{(1)}:=\min\{f_i,\min_{e\in p_i,e\in E_1}\{\tilde{C}_A(e)\}\}.
%	\]
%	\item \textbf{Phase II:} For $i\in\{1,\ldots,k\}: $
%	\[
%	\tilde{f}_{iA}^{(2)}:=\tilde{f}_{iA}^{(1)}\cdot\prod_{e\in p_i,e\in E_2,\tilde{C}_A(e)\le \sum_{p_j\ni e}f_j}\frac{\tilde{C}_A(e)}{\sum_{p_j\ni e}f_j}.
%	\]
%	\item Set $\bar{\Lambda}(A,P):=\sum_i(f_i-\tilde{f}_{iA}^{(2)})$.
%\end{itemize}

%It is straightforward to show that $\{\tilde{f}_{iA}^{(2)}\}$'s satisfy the constraints (\ref{capacity}) and (\ref{utility}), and thus are feasible to the maximization problem (\ref{throughput}). Hence, $\sum_i \tilde{f}_{iA}^{(2)}\le T(A,P)$ and $\bar{\Lambda}(A,P)\ge {}(A,P)$, which shows that \bar{\Lambda} is an upper bound of $\Lambda$.

%Similar as $\Lambda$, \bar{\Lambda} can also be viewed as a set function defined on the subsets of $E$ with $\bar{\Lambda}(\mathbf{f},P)=(E_\mathbf{f},P)$. 
Substituting $\Lambda$ with $\bar{\Lambda}$ in \textbf{Algorithm \ref{algorithm:greedy}}, we obtain the Recursive Greedy algorithm for the case of non-disjoint user paths. We will refer to this algorithm as the \textbf{Extended Recursive Greedy algorithm}. The name is justified by noting that when the user paths are disjoint, $E_2=\emptyset$ and $\bar{\Lambda}=\Lambda$, the Extended Recursive Greedy algorithm degenerates to \textbf{Algorithm \ref{algorithm:greedy}}.

Before analyzing the performance of the algorithm,
we establish two lemmas 
that show the monotonicity and submodularity of $\bar{\Lambda}$, and bound the gap between $\bar{\Lambda}$ and $\Lambda$, respectively. The proofs of the lemmas are given in Appendix \ref{app:prooflemma}
\begin{lemma}\label{lemma:extendproperties}
	Consider $\bar{\Lambda}(\cdot,P):2^E\mapsto \mathbb{R}^*$ as a set function. $\bar{\Lambda}(\cdot,P)$ is monotone and submodular.
\end{lemma}

%%%pan out this lemma
%\begin{proof}
%	See Appendix \ref{appendix:prooflemma}.
%\end{proof}

%The second lemma bounds the gap between \bar{\Lambda} and $\Lambda$.
\begin{lemma}\label{lemma:gap}
	$\Lambda(A,P)\le \bar{\Lambda}(A,P)\le (b+1)\cdot\Lambda(A,P)$ for all $A\subseteq E$, where $b=\max_i{|E_2\cap p_i|}$,\footnote{$|A|$ denotes the cardinality of set $A$} i.e., the maximum number of edges that a user path shares with other user paths.   
\end{lemma}
%\begin{proof}
%	See Appendix \ref{appendix:prooflemma}.
%\end{proof}

%%%This needs to be modified accordingly!
Now, we are ready to analyze the performance of the Extended Recursive Greedy algorithm.
\begin{theorem}\label{thm:recgreedy2}
	If $I\ge \lceil\log d\rceil$, then the Extended Recursive Greedy algorithm returns an $s$-$t$ path $E_\mathbf{f}$ that satisfies \[\Lambda(E_\mathbf{f},P)\ge \frac{1}{(b+1)\cdot(\lceil\log d\rceil+1)}\Lambda(E_\mathbf{f^*},P),\] where $d$ is the length of $E_\mathbf{f^*}$ and $b=\max_i|E_2\cap p_i|$.
	%in $O\left(n^{ai+1}m\right)$ running time, where $E_\mathbf{f^*}$ is the $s$-$t$ path corresponding to an optimal single-path flow, $d$ is the length of $E_\mathbf{f^*}$ and $b=\max_i|E_2\cap p_i|$.
\end{theorem} 
\begin{proof}
	By Lemma \ref{lemma:extendproperties} and Theorem \ref{thm:recgreedy}, we have $\bar{\Lambda}(E_\mathbf{f},P)\ge \frac{1}{(\lceil\log d\rceil+1)}\bar{\Lambda}(E_\mathbf{f^*},P)$ when $I\ge \lceil\log d\rceil$. Invoking Lemma \ref{lemma:gap}, we obtain that
	\begin{small}
	\begin{align*}
	 \Lambda(E_\mathbf{f},P)&\ge \frac{1}{b+1}\bar{\Lambda}(E_\mathbf{f},P) \ge \frac{1}{(b+1)\cdot(\lceil\log d\rceil+1)}\bar{\Lambda}(E_\mathbf{f^*},P)\\ &\ge \frac{1}{(b+1)\cdot(\lceil\log d\rceil+1)}\Lambda(E_\mathbf{f^*},P),
	 \end{align*}
	 \end{small}
	  which concludes the proof.
\end{proof}
	  
	  Note that the computation of $\bar{\Lambda}$ takes $O(m)$ time. 
	Therefore, taking $I=\log n$ we get a $\frac{1}{(b+1)\cdot(\lceil\log d\rceil+1)}$-approximation algorithm with a time complexity of $O((2n)^{\log n}m)$. Although in the worst case, $b$ can be at the same order as $n$.  In the cases, $b$ is of $O(\log n)$,  and the Extended Recursive Greedy algorithm still enjoys a logarithmic approximation ratio.

\section{Robust Flow Interdiction} \label{sec:robust}
In this section, we investigate the robust flow interdiction problem. Following the road map of deterministic flow interdiction, we first describe the formal definition of the problem, then show its computational complexity, and finally present the approximation framework for the problem.

\subsection{Problem Formulation}
While deterministic flow interdiction considers the case where the interdictor has definitive knowledge of the user paths, robust flow interdiction concerns scenarios where such knowledge is not available.
 We model this more complicated situation using the robust optimization framework \cite{cite:robust}. Instead of having certain knowledge of $P$, the interdictor only knows that $P$ lies in an uncertainty set $\mathcal{U}=\{P_1,\ldots,P_\xi\}$. Each $P_l=\{p_{l1},\ldots,p_{lk_l}\}\in\mathcal{U}$, associated with initial flow values $\{\lambda_{l1},\ldots,\lambda_{lk_l}\}$, is a candidate set of paths that the users are operating on. The interdictor aims to hedge against the worst case, maximizing the minimum throughput reduction achieved over all candidates $P$.
\begin{definition}[Robust Flow Interdiction] \label{def:robust}
	Given the uncertain set $\mathcal{U}=\{P_1,\ldots,P_\xi\}$ of the user paths and the associated initial flow values on user paths for each $P\in\mathcal{U}$, the robust flow interdiction problem seeks an interdiction strategy $w$ that maximizes the worst case throughput reduction, i.e., $w\in\arg\max_{w'}\min_{P\in\mathcal{U}}\Lambda(w',P)$. 
\end{definition}
\textbf{Example:}
As an example of the robust flow interdiction problem, we again consider the network in Figure \ref{fig:example}. The interdictor has source $s$, detination $t$ and budget $\gamma=2$. Assume that the interdictor only knows that the users are sending flow on either $\{p_1,p_2\}$ or $\{p_1,p_3\}$, and the initial flow values on $p_1,p_2,p_3$ are all three.  This corresponds to the robust flow interdiction with $\mathcal{U}=\{\{p_1,p_2\},\{p_1,p_3\}\}$. Let $\mathbf{f}_1$ be the $s$-$t$ flow such that $\mathbf{f}_1(s,v_3)=\mathbf{f}_1(v_3,v_4)=\mathbf{f}_1(v_4,t)=2$ and $\mathbf{f}_2$ be the $s$-$t$ flow such that $\mathbf{f}_2(s,v_1)=\mathbf{f}_2(v_1,v_3)=\mathbf{f}_2(v_3,v_2)=\mathbf{f}_2(v_2,t)=2$. The optimal strategy in this case is $w(\mathbf{f}_1)=1/3,w(\mathbf{f}_2)=2/3$, and the worst case throughput reduction equals $8/3$ as $\Lambda(w,\{p_1,p_2\})=\Lambda(w,\{p_1,p_3\})=8/3$. Note that in this example, no pure interdiction strategy can achieve a worst case throughput reduction of $8/3$, which demonstrates the superiority of mixed strategies in the robust flow interdiction setting.

The robust flow interdiction problem subsumes the deterministic one as a special case by setting $\mathcal{U}=\{P\}$. Therefore, we immediately have the following proposition.
\begin{proposition}\label{proposition:hardness1}
	The robust flow interdiction problem is NP-hard.
\end{proposition}

Before presenting our approximation framework, we present a linear programming (LP) formulation that serves as an alternative solution to the robust flow interdiction problem. According to Proposition \ref{proposition:property}, we can restrict our attention to distributions on the set of single-path flows with value $\gamma$. Therefore, in the following, \textit{the distributions we refer to are all on the set of single-path flows in $\mathcal{F}_{\gamma}$}. We enumerate such single-path flows in an arbitrary order and associate with each single-path flow $\mathbf{f}_i$ a variable $w_i$. Consider the linear program:
 \begin{align}
\quad\text{maxi}&\text{mize } \textstyle z\label{LP:robust}\\
\text{\textbf{s.t. }}& \sum_iw_i\Lambda(\mathbf{f}_i,P)\ge z,\quad \forall P\in \mathcal{U}\nonumber\\
& \sum_iw_i=1\nonumber\\
&\ w_i\ge 0,\quad \forall i \nonumber
\end{align}

Clearly, the solution to the LP corresponds to an optimal interdiction strategy $w$ to the robust flow interdiction problem with $w(\mathbf{f}_i)=w_i$. Hence, formulating and solving the LP is a natural algorithm for the robust flow interdiction. However, as the number of single-path flows can be exponential in the number of nodes $n$, the LP may contain an exponential number of variables. It follows that the algorithm has an undesirable exponential time complexity. We use this algorithm in the simulations to obtain optimal interdiction strategies for comparisons with our approximation framework. 
Another issue arises when the number of single-path flows is exponential in the number of nodes, that is, even outputting the strategy $w$ takes exponential time. This makes it impractical and unfair to compare any sub-exponential time approximation procedure to the optimal solution. We get around this issue by comparing our solution to the optimal interdiction strategy that takes non-zero values on at most $N_0$ single-path flows, where $N_0$ is a pre-specified number  bounded by some polynomial of $n$. We refer to such strategies as $N_0$-bounded strategies. The optimal $N_0$-bounded strategy corresponds to the best strategy that uses at most $N_0$ different interdicting flows. Note that such restriction does not trivialize the problem since we place no limitation on the set but just the number of single-path flows that the interdictor can use.

%%%Think carefully, how to clearly deliver the ideas!!!!
\subsection{Approximation Framework}
In this section, we present the approximation framework we propose for the robust flow interdiction problem. As a generalization of the deterministic version, the robust flow interdiction is more complicated since it involves maximizing the minimum of a set of functions. The (Extended) Recursive Greedy algorithm cannot be directly adapted to this case. Instead, we design an approximation framework that integrates the Extended Recursive Greedy algorithm as a sub-procedure. The framework only incurs a logarithmic loss in terms of approximation ratio. 
 
The description and analysis of the approximation framework are carried out in three steps. First, we justify that it is sufficient to consider the robust flow interdiction problem with parameters taking rational/integral values. In the second step, building on the rationality/integrality of parameter values, we convert the problem to a sequence of integer linear programs. Finally, we solve the sequence of integer programs through iteratively invoking the Extended Recursive Greedy algorithm.

\subsubsection{Rationalizing the Parameters}
In the first step, we show that not much is lost if we only consider the interdiction strategies that take rational values and restrict the throughput reduction to take integer values. Specifically, let $N=N_0^2+N_0$ and $\mathbb{Q}_N=\{\frac{\beta}{N}:\beta\in\mathbb{N},0\le\beta\le N\}$ be the set of non-negative rational numbers that can be represented with $N$ as denominator. Further, we define $\mathcal{W}_N$ to be the set of strategies that take value in $\mathbb{Q}_N$, i.e., $\mathcal{W}_N=\{w:\mathcal{F}_\gamma\mapsto\mathbb{Q}_N, \sum_\mathbf{f}w(\mathbf{f})=1\}$. We use $w^*$ to represent the optimal $N_0$-bounded interdiction strategy, and $w_N^*$ to represent optimal strategy in $\mathcal{W}_N$.  The following lemma states that $w^*$ can be well approximated by $w_N^*$.

\begin{lemma}\label{lemma:rationalize}
	%$w_N^*$ is a $N$-bounded strategy,
	 For all $P\in\mathcal{U}$, $\Lambda(w^*_N,P)\ge \frac{N_0}{N_0+1}\Lambda(w^*,P)$.
\end{lemma}
\begin{proof}
Consider $\tilde{w}^*$ such that $\tilde{w}^*(\mathbf{f})=\frac{\lceil N_0^2w^*(\mathbf{f})\rceil}{N_0^2+N_0}$ for all $w^*(\mathbf{f})>0$ and $\tilde{w}^*(\mathbf{f})=0$ otherwise. Since $w^*$ is a $N_0$-bounded strategy, $\sum_{\mathbf{f}} \tilde{w}^*(\mathbf{f})\le \frac{N_0+N_0^2\sum_\mathbf{f}w^*(\mathbf{f})}{N_0^2+N_0}= 1$. Hence, we can augment $\tilde{w}^*$ into a strategy in $\mathcal{W}_N$ by adding $1-\sum_\mathbf{f'}\tilde{w}^*(\mathbf{f}')$ to some $\tilde{w}^*(\mathbf{f})$. With a little abuse of notation, we use $\tilde{w}^*$ to denote the resulting strategy. By the definition of $\tilde{w}^*$, we have
\begin{small}
	\begin{align*}
	\sum_{\mathbf{f}}\tilde{w}^*(\mathbf{f})\Lambda(\mathbf{f},P)\ge \frac{N_0^2}{N_0^2+N_0}\sum_{\mathbf{f}}w^*(\mathbf{f})\Lambda(\mathbf{f},P)=\frac{N_0}{N_0+1}\Lambda(w^*,P).
	\end{align*}
\end{small}
As $\tilde{w}^*\in\mathcal{W}_N$, it follows that $\Lambda(w_N^*,P)\ge \Lambda(\tilde{w}^*,P)=\sum_{\mathbf{f}}\tilde{w}^*(\mathbf{f})\Lambda(\mathbf{f},P)\ge \frac{N_0}{N_0+1}\Lambda(w^*,P)$.
\end{proof}

We now proceed to argue that it suffices to consider the throughput reduction function $\Lambda$ to take integral values that are bounded by some polynomial of $n$. First, when the integrality of $\Lambda$ is not satisfied, we can always use standard scaling and rounding tricks to get a new instance of the problem, where $\Lambda$ takes integral values. Our framework can be applied to the new instance, yielding an interdiction strategy that has almost the same performance guarantee for both the original and the new instances. We defer the formal statement and proof of this to Appendix \ref{app:integrality}, as it involves definitions in subsequent sections. Second, since $\gamma$ is bounded by some polynomial of $n$, $\max_{w,P}\Lambda(w,P)$ is also bounded by some polynomial of $n$.
 Now, let $M=N\max_{w,P}\Lambda(w,P)$. We can thus without loss of generality assume that $M$ is an integer bounded by some polynomial of $n$.

With the above results, we move into the second step, that converts the robust flow interdiction problem into a sequence of integer linear programs.

%%Scale and round, add on this lemma
%%%Explain the connection between the sequence of the programs and the robust flow interdiction.
%%%Introduce multiset multicover problem in detail

%%%Try drop the multiset multicover thing.
\subsubsection{Converting into Integer Linear Programs}
Recall the enumeration of single-path flows in the LP (\ref{LP:robust}). This time, we associate each flow $\mathbf{f}_i$ with an integral variable $x_i$. Consider the following integer program $ILP(\kappa)$ parameterized by a positive integer $\kappa\le M$.
 \begin{align}
 \quad\text{minimize } & \textstyle\sum_{i}x_i\label{multiset-mutlicover}\\
 \text{\textbf{s.t. }}& \sum_ix_i\Lambda(\mathbf{f}_i,P)\ge \kappa,\quad \forall P\in \mathcal{U} \label{cover}\\
 &\ x_i\in\mathbb{N},\quad \forall i \label{integrality}
 \end{align}
Each $x_i$ indicates the number of times $\mathbf{f}_i$ is selected. $ILP(\kappa)$ can be interpreted as selecting the single-path flows for the minimum total number of times that achieve a throughput reduction of $\kappa$ for all candidate $P$.
%For each $\kappa$, $Cover(\kappa)$ can be interpreted as a \textit{multiset-multicover} problem. Given a universe of elements and a collection
%of multisets, the multiset multicover problem is to cover all elements for at least a number of times as specified in their coverage requirements with the minimum number of multisets \cite{cite:MCMC}.
% In our case, each $\mathbf{f}_i$ can be considered as a multiset, each $P\in\mathcal{U}$ is an element to cover. Each $\mathbf{f}_i$ covers $P$ for $(\mathbf{f}_i,P)$ times and the coverage requirement for each $P$ is $\kappa$. For each $i$, $x_i$ indicates the number of times we choose $\mathbf{f}_i$ and $\sum_ix_i$ equals the cardinality of the computed collection.

 For each $\kappa$, we denote by $N_\kappa$ the optimal value of $ILP(\kappa)$. If we can obtain an optimal solution $\{x\}$ to $ILP(\kappa)$, then the strategy $w$ with $w(\mathbf{f}_i)=x_i/N_{\kappa}$ satisfies $\min_{P\in\mathcal{U}}\Lambda(w,P)\ge \kappa/N_{\kappa}$. In the following lemma, we show that the strategy constructed according to the solution to the integer program with the maximum value of $\kappa/N_{\kappa}$ is a close approximation to the optimal $N_0$-bounded strategy in terms of worst case throughput reduction. 
 %Link between strategy and everything 
%For any solution $\{x\}$ to some integer program $Cover(\kappa)$, we define the interdiction strategy induced by $\{x\}$ as the one such that $w(\mathbf{f}_j)=x_j/\sum_ix_i$ for all $\mathbf{f}_j$.

\begin{lemma}\label{lemma:multiset}
	 %Let $\{x^*\}$ be the solution to $Cover(\kappa^*)$ with $\kappa^*=\arg\max_{1\le\kappa\le M}({\kappa}/{N_\kappa})$. 
	 %Then the strategy $w_{\kappa^*}$ induced by $\{x^*\}$ is an $\xi M$-bounded strategy such that 
	 Let $\kappa^*=\arg\max_{1\le\kappa\le M}({\kappa}/{N_\kappa})$. We have
	 $\frac{\kappa^*}{N_{\kappa^*}}\ge \min_{P\in\mathcal{U}}\Lambda(w_N^*,P)\ge \frac{N_0}{N_0+1}\min_{P\in\mathcal{U}}\Lambda(w^*,P)$. 
\end{lemma}
\begin{proof}
Define $\kappa'$ to be $\min_{P\in\mathcal{U}}\sum_{\mathbf{f}}Nw^*_N(\mathbf{f})\Lambda(\mathbf{f},P)=\min_{P\in\mathcal{U}}N\Lambda(w^*_N,P)$. Note that $\kappa'$ is a positive integer and $\kappa'\le M$. Thus, by the definition of $\kappa^*$, we have $\kappa'/N_{\kappa'}\le \kappa^*/N_{\kappa^*}$. 
%Now, we consider the optimal solution $\{x'\}$ to the integer program $Cover(\kappa')$. By the feasibility of $\{x'\}$, we have $\forall P\in\mathcal{U},\ \sum_ix'_i(\mathbf{f}_i,P)\ge\kappa'$. 
Also, observe that the solution $\{x\}$ with $x_i=Nw^*_N(\mathbf{f}_i)$ is feasible to $ILP(\kappa')$. Therefore,
$ N_{\kappa'}\le\sum_iw^*_N(\mathbf{f}_i)N= N.$
It follows that
%\begin{small}
	\[
	\frac{\kappa^*}{N_{\kappa^*}}\ge \frac{\kappa'}{N_{\kappa'}}\ge \frac{\kappa'}{N}\nonumber
	=\min_{P\in\mathcal{U}}\Lambda(w^*_N,P)\ge \frac{N_0}{N_0+1}\min_{P\in \mathcal{U}}\Lambda(w^*,P).
	\]	
%\end{small}
\end{proof}

Connecting the analysis so far, we have a clear procedure to compute a near-optimal interdiction strategy for the robust flow interdiction. First, we construct and solve $ILP(\kappa)$ for $1\le\kappa\le M$. Second, we take optimal solution with the maximal $\kappa/N_{\kappa}$ and obtain its corresponding interdiction strategy, which is within a factor of $\frac{N_0}{N_0+1}$ to the optimal $N_0$-bounded strategy. The final step of our framework is devoted to solving $ILP(\kappa)$.

%%%%Continue to fix this section.
\subsubsection{Solving the Integer Linear Programs}\vspace{1mm}
Resembling (\ref{LP:robust}), each $ILP(\kappa)$ involves potentially exponential number of variables. What is different and important is that, we can obtain a $\frac{1}{\log M}$-approximation through a greedy scheme that iteratively chooses a single-path flow according to the following criterion:
%through a greedy scheme that chooses the multiset that covers the most number of uncovered elements (with multiplicity counted) at every iteration \cite{cite:multiset}.
let $\{x\}$ indicate the collection of flows that have been chosen so far, i.e., each $\mathbf{f}_i$ has been chosen for $x_i$ times. Let $i^*$ be 
%At the current iteration, each element $P\in\mathcal{U}$ has been covered $\sum_jx_j\Lambda(\mathbf{f}_j,P)$ times. 
%The greedy scheme works as:
%\begin{itemize}
%\item Let $i^*$ be
\begin{small}
	\begin{align} \arg\max_i\sum_{P\in\mathcal{U},\kappa\ge\sum_jx_j\Lambda(\mathbf{f}_j,P)}\min\{\kappa-\sum_jx_j\Lambda(\mathbf{f}_j,P),\Lambda(\mathbf{f}_i,P)\}. \label{eq:greedyoracle}
\end{align}
\end{small}
The greedy scheme chooses $\mathbf{f}_{i^*}$ at the current iteration and increments $x_{i^*}$ by 1. The above procedure is repeated until we have $\sum_ix_i\Lambda(\mathbf{f}_i,P)\ge k$ for all $P\in\mathcal{U}$. Moreover, if we apply an $\alpha$-approximate greedy scheme, which chooses $\mathbf{f}_i$ that is an $\alpha$-optimal solution to (\ref{eq:greedyoracle}), then the final solution we obtain is $\alpha\log M$-optimal.
Essentially, (\ref{eq:greedyoracle}) selects the flow that provides the maximum marginal gain with respect to satisfying the constraints (\ref{cover}) for all $P\in\mathcal{U}$. That the ($\alpha$-approximate) greedy scheme achieves an logarithmic approximation follows from the relation of $ILP(\kappa)$ to the multiset-multicover problems and the results therein \cite{cite:multiset}, which we omit here due to space limitation. 
 Now recall the equivalence between $\Lambda(\mathbf{f},P)$ and $\Lambda(E_\mathbf{f},P)$ established in Section \ref{sec:deterministic}. We proceed to show that the Recursive Greedy algorithm can be used to construct an approximate greedy scheme. First, we have the following lemma.
 
 \begin{lemma}\label{lemma:robustsubmodular}
 	If $\Lambda$ is monotone and submodular, then the objective function of (\ref{eq:greedyoracle}) is also monotone and submodular.
 \end{lemma}
\begin{proof}
	Note that at any iteration, $\sum_jx_j\Lambda(\mathbf{f}_j,P)$ is a known constant. Hence, for each $P$, $\min\{\kappa-\sum_jx_j\Lambda(\mathbf{f}_j,P),\Lambda(\mathbf{f}_i,P)\}$ is the minimum of a constant and a monotone submodular function, which is also monotone and submodular. It follows that the objective function of (\ref{eq:greedyoracle}) is monotone and submodular.
\end{proof}
%  Projecting into our problem, at each iteration, for all $P\in\mathcal{U}$, denoting the value we have covered (i.e., achieved by the single-path flows that have been selected) as $\kappa_P$, we choose a single path flow $\mathbf{f}$ that maximizes
%\begin{align}
%\sum_{P\in\mathcal{U}}\min\{\kappa-\kappa_P,(\mathbf{f},P)\}. \label{solvemulti}
%\end{align}

%Since the minimum of a monotone submodular function and a constant is monotone and submodular, and the sum of monotone submodular functions is still monotone and submodular, the objective function (\ref{solvemulti}) is a monotone and submodular set function. 

By Lemma \ref{lemma:robustsubmodular}, the Recursive Greedy algorithm (or the Extended Recursive Greedy algorithm using $\bar{\Lambda}$ instead of $\Lambda$ when the user paths are not disjoint) can be applied to the maximization of (\ref{eq:greedyoracle}) and enjoys the same performance guarantee as in Theorems \ref{thm:recgreedy} and \ref{thm:recgreedy2}. Hence, the final step can be completed by an approximate greedy scheme that iteratively invokes the (Extended) Recursive Greedy algorithm. We now summarize the three steps of our approximation framework for the robust flow interdiction as \textbf{Algorithm \ref{algorithm:robust}} and analyze its performance.

	\begin{algorithm}
	\caption{Algorithm for the Robust Flow Interdiction}
	\begin{algorithmic}[1]	\label{algorithm:robust}		
		\REQUIRE{Network graph $G$, Uncertainty set $\mathcal{U}=\{P_1,\ldots,P_\xi\}$, Interdictor's source $s$, destination $t$ and budget $\gamma$ }\\
		\ENSURE{Interdiction Strategy $w$}
		\STATE Formulate $ILP(\kappa)$ for $1\le\kappa\le M$.
		\STATE Solve each $ILP(\kappa)$ using the approximate greedy scheme based on the (Extended) Recursive Greedy algorithm.
		\STATE Take the solution $\{x\}$ to $ILP(\kappa)$ with the maximum value of $\kappa/\sum_jx_j$ and construct $w$ by setting $w(\mathbf{f}_i)=x_i/\sum_jx_j$ for all $i$.
		\RETURN{$w$}
	\end{algorithmic}
\end{algorithm}

\begin{theorem}\label{thm:robust}
\textbf{Algorithm \ref{algorithm:robust}} returns an interdiction strategy $w$ that satisfies
\begin{small}
	\begin{align*}
 &\min_{P\in\mathcal{U}}\Lambda(w,P)\\&\ge \left(\frac{N_0}{(N_0+1)(b+1)\log M\cdot(\lceil\log d\rceil+1)}\right)\min_{P\in\mathcal{U}}\Lambda(w^*,P),
 \end{align*}
\end{small} where $w^*$ is the optimal $N_0$-bounded strategy. 
%The time complexity of \textbf{Algorithm \ref{algorithm:robust}} is $O\left(M^2mn^{ai}\right)$.
\end{theorem}
\begin{proof}
	Let $ILP(\kappa)$ and $\{x\}$ be the integer linear program and its solution that correspond to $w$. We inherit the definition of $\kappa^*$ in Lemma \ref{lemma:multiset} and further define $\{x^*\}$ to be the solution that \textbf{Algorithm \ref{algorithm:robust}} computes for $ILP(\kappa^*).$  We have
%	\begin{small}
%	\begin{align}
%	&\min_{P\in\mathcal{U}}\Lambda(w,P)=\min_{P\in \mathcal{U}}\sum_iw(\mathbf{f}_i)\Lambda(\mathbf{f}_i,P)\nonumber\\
%	&=\min_{P\in \mathcal{U}}\sum_i\frac{x_i}{\sum_jx_j}\Lambda(\mathbf{f}_i,P)\ge \frac{\kappa}{\sum_jx_j}\ge \frac{\kappa^*}{\sum_jx_j^*}\nonumber\\
%	&\ge \left(\frac{1}{(b+1)\log M\cdot(\lceil\log d\rceil+1)}\right)\frac{\kappa^*}{N_{\kappa^*}} \label{ieq:theorem3_1}\\
%	&\ge \left(\frac{1}{(b+1)\log M\cdot(\lceil\log d\rceil+1)}\right)\min_{P\in\mathcal{U}}\Lambda(w_N^*,P)\label{ieq:theorem3_2}\\
%	&\ge  \left(\frac{N_0}{(N_0+1)(b+1)\log M\cdot(\lceil\log d\rceil+1)}\right)\min_{P\in\mathcal{U}}\Lambda(w^*,P)\label{ieq:theorem3_3},
%	\end{align}
%	\end{small}

	\begin{small}
	\begin{align}
	&\min_{P\in\mathcal{U}}\Lambda(w,P)=\min_{P\in \mathcal{U}}\sum_iw(\mathbf{f}_i)\Lambda(\mathbf{f}_i,P)\nonumber\\
	&=\min_{P\in \mathcal{U}}\sum_i\frac{x_i}{\sum_jx_j}\Lambda(\mathbf{f}_i,P)\ge \frac{\kappa}{\sum_jx_j}\ge \frac{\kappa^*}{\sum_jx_j^*}\nonumber\\
	&\ge \left(\frac{1}{(b+1)\log M\cdot(\lceil\log d\rceil+1)}\right)\frac{\kappa^*}{N_{\kappa^*}} \label{ieq:theorem3_1}\\
	&\ge  \left(\frac{N_0}{(N_0+1)(b+1)\log M\cdot(\lceil\log d\rceil+1)}\right)\min_{P\in\mathcal{U}}\Lambda(w^*,P)\label{ieq:theorem3_3},
	\end{align}
\end{small}
%	where inequality (\ref{ieq:theorem3_1}) follows from Theorem \ref{thm:recgreedy2} and the results in \cite{cite:multiset}, inequality (\ref{ieq:theorem3_2}) follows from Lemma \ref{lemma:multiset} and inequality (\ref{ieq:theorem3_3}) follows from Lemma \ref{lemma:rationalize}.
	where inequality (\ref{ieq:theorem3_1}) follows from Theorem \ref{thm:recgreedy2} and the results in \cite{cite:multiset}, and inequality (\ref{ieq:theorem3_3}) follows from Lemma \ref{lemma:multiset}.
\end{proof}
	
	\textbf{Time Complexity:} Note that \textbf{Algorithm \ref{algorithm:robust}} solves $M$ integer linear programs, and it takes at most $M\xi$ calls of the (Extended) Recursive Greedy algorithm for each program since the number of iterations is bounded by $M\xi$, where $\xi=|\mathcal{U}|$. Furthermore, at the third step, there are at most $N_\kappa\le M\xi$ variables with non-zero values in the solution $\{x\}$, which implies that $w$ can be output in $O(M\xi)$ time. Therefore, the time complexity of \textbf{Algorithm \ref{algorithm:robust}} is $O\left(m(M\xi)^2(2n)^{\log n}\right)$.

%\textbf{Remark:} As in the deterministic case, we can set $a=2$ and get a pseudo-polynomial time algorithm with poly-logarithmic approximation guarantee.

\section{Simulations} \label{sec:simulations}
In this section, we present our evaluation of the performance of the proposed algorithms. We first introduce the simulation environment in the following and then show the detailed results in subsequent sections.
\subsection{Simulation Setting}
%\subsubsection{Simulation Dataset}
We adopt the Gnutella peer to peer network data set from \cite{cite:simulation1}. We extract 20 networks of 1000 nodes, %and randomly choose 5 connected node pairs in each network as the source and destination of the interdictor, thus forming 100 simulation scenarios in total.
and make the networks acyclic by removing a minimal feedback edge set from each of them. The capacities of the edges are sampled from a normal distribution with mean 20 and standard deviation 3. The budget of the interdictor is set to the minimum capacity of the edges in each network. 
\subsection{Deterministic Flow Interdiction}
In the deterministic flow interdiction, we divide our simulations into two parts, where the user paths are disjoint and non-disjoint respectively. In the first part, we designate $k$ disjoint paths in each network as user paths with $k$ varying in $\{10,20,\ldots,100\}$. In the second part, we follow the similar route, except that the user paths are randomly chosen without guaranteeing their disjointness. For each network, we randomly select five connected node pairs as the source and destination of the interdictor. Thus, for each number of user paths, we have 100 simulation scenarios in total (20 networks times 5 $s$-$t$ pairs).

\subsubsection{Algorithms Involved in Performance Comparisons} We apply the Recursive Greedy algorithm when the user paths are disjoint and run the extended one when the user paths are non-disjoint. We vary the recursion depth, i.e., the value of $I$ in \textbf{Algorithm \ref{algorithm:greedy}} to evaluate its influence on the algorithms' performance.  Our algorithms are compared to a \textit{brute force} algorithm that enumerates all the paths between the interdictor's source and destination, which computes the optimal interdiction strategy.
\subsubsection{Performance Metric}
We calculate the ratio of the throughput reduction of the interdiction strategies by our algorithms to that of the optimal solutions obtained by the brute force algorithm. The results reported are the average over all the 100 scenarios.
\subsubsection{Simulation Results}
We plot the results of our algorithms on deterministic flow interdiction with disjoint and non-disjoint user paths in Figures \ref{fig:disjoint1000} and \ref{fig:non_disjoint1000}.

From Figure \ref{fig:disjoint1000}, we can see that: (i). by setting the recursion depth to two, we get interdiction strategies with throughput reduction more than 90\% of the optimal (0.9-approximation) and (ii). by setting the recursion depth to three, we recover the optimal interdiction strategies. Furthermore, in the simulations, we find that when the recursive depth is three, the number of paths examined by the Recursive Greedy algorithm is just about one fifth of the total number of $s$-$t$ paths. This suggests that the typical performance and running time are even better than what the theoretical analysis predicts. Finally, we observe that, in general, our algorithms perform better when the number of user paths is large. This observation also holds in subsequent cases. One possible explanation for this is that more user paths present more opportunities for throughput reduction, making (near-)optimal interdicting flows easier to find.

As demonstrated in Figure \ref{fig:non_disjoint1000}, the deterministic flow interdiction is harder to approximate when the user paths are non-disjoint. But we can still get 0.8-approximations with a recursion depth of three and 0.95-approximations with a recursion depth of four. Also, though we have not plotted in the figure, we have seen that increasing the recursion depth to five or six does not further improve the performance. Therefore, the gap between the Extended Recursive Greedy algorithm with depth of four and the optimal can be attributed to the loss brought by the approximate throughput reduction function $\bar{\Lambda}$.
%%path 越多，表现越好
%%depth 4/3 就够了

%\begin{figure}
%	\centering
%	\includegraphics[width=0.9\linewidth,height = 0.18\textheight]{Simulations/Figures/disjoint_1000_1}
%	\caption{Ratio of throughput reduction of the solutions by the Recursive Greedy algorithm with different recursion depths to the optimal.}
%	\label{fig:disjoint1000}
%\end{figure}

%\begin{figure}
%	\centering
%	\includegraphics[width=0.9\linewidth,height = 0.18\textheight]{Simulations/Figures/non_disjoint_1000_1}
%	\caption{Ratio of throughput reduction of the solutions by the Extended Recursive Greedy algorithm with different recursion depths to the optimal.}
%	\label{fig:non_disjoint1000}
%\end{figure}

\begin{figure*}[htbp]
\centering
\subfigure[]{
	\begin{minipage}[]{0.32\linewidth}
		\centering
		\vspace{-4mm}
		\includegraphics[width=1.02\linewidth]{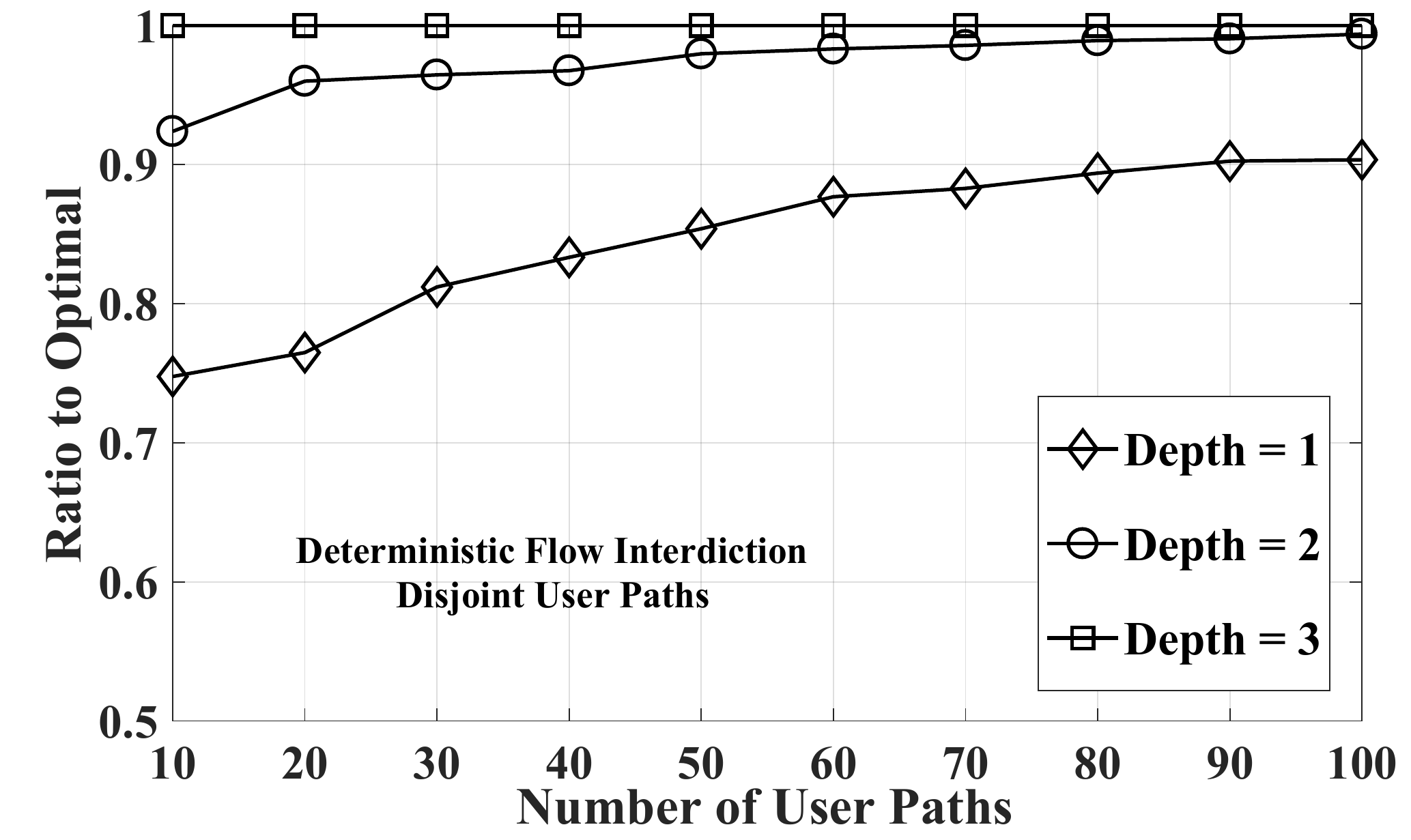}
		%\caption{fig1}
		\vspace{-4mm}
		\label{fig:disjoint1000}
	\end{minipage}%
	
}
\subfigure[]{
	\begin{minipage}[]{0.32\linewidth}
		\centering
		\vspace{-4mm}
		\includegraphics[width=1.02\linewidth]{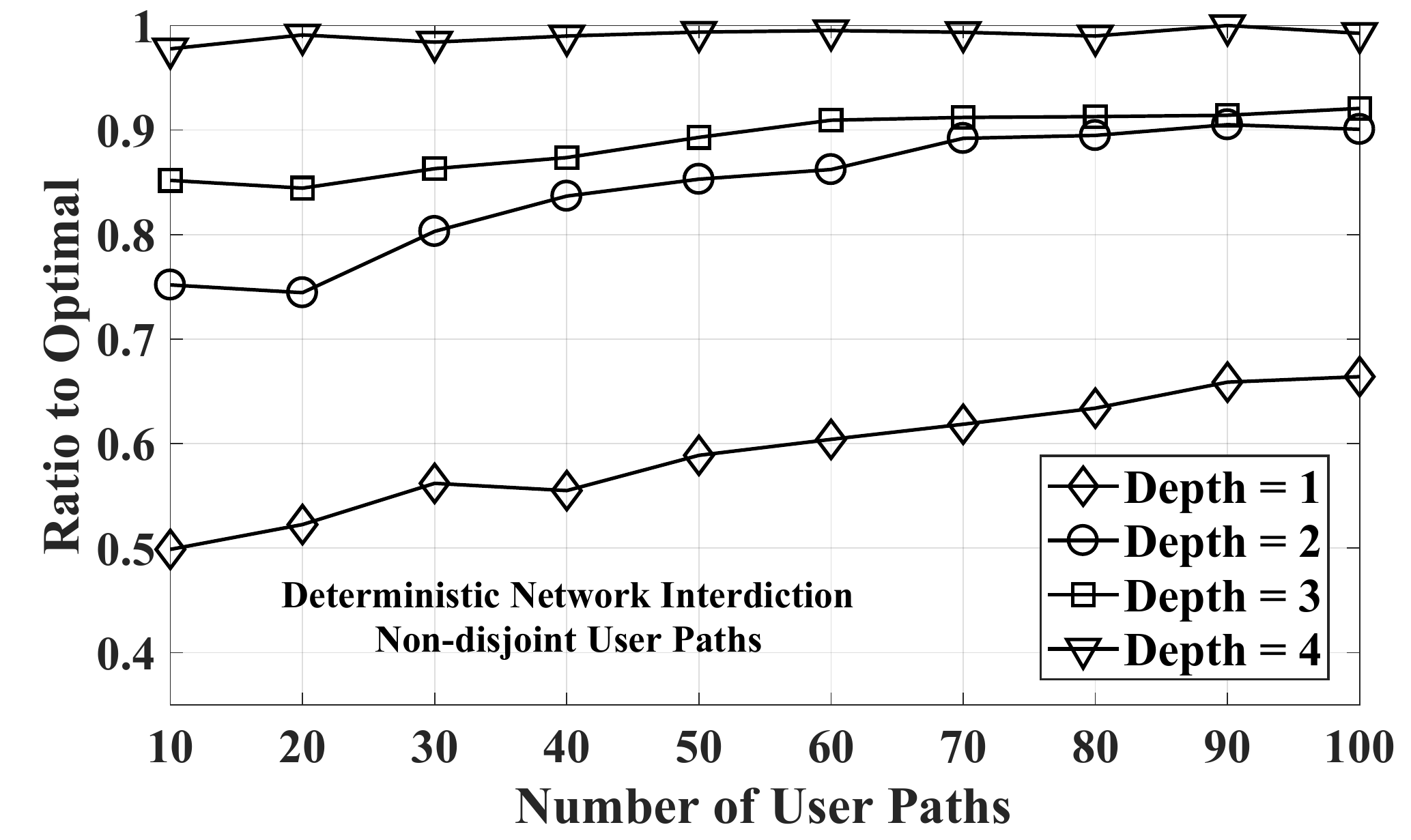}
		\vspace{-4mm}
		\label{fig:non_disjoint1000}
	\end{minipage}%
	
}
\subfigure[]{
	\begin{minipage}[]{0.32\linewidth}
		\centering
		\vspace{-4mm}
		\includegraphics[width=1.02\linewidth]{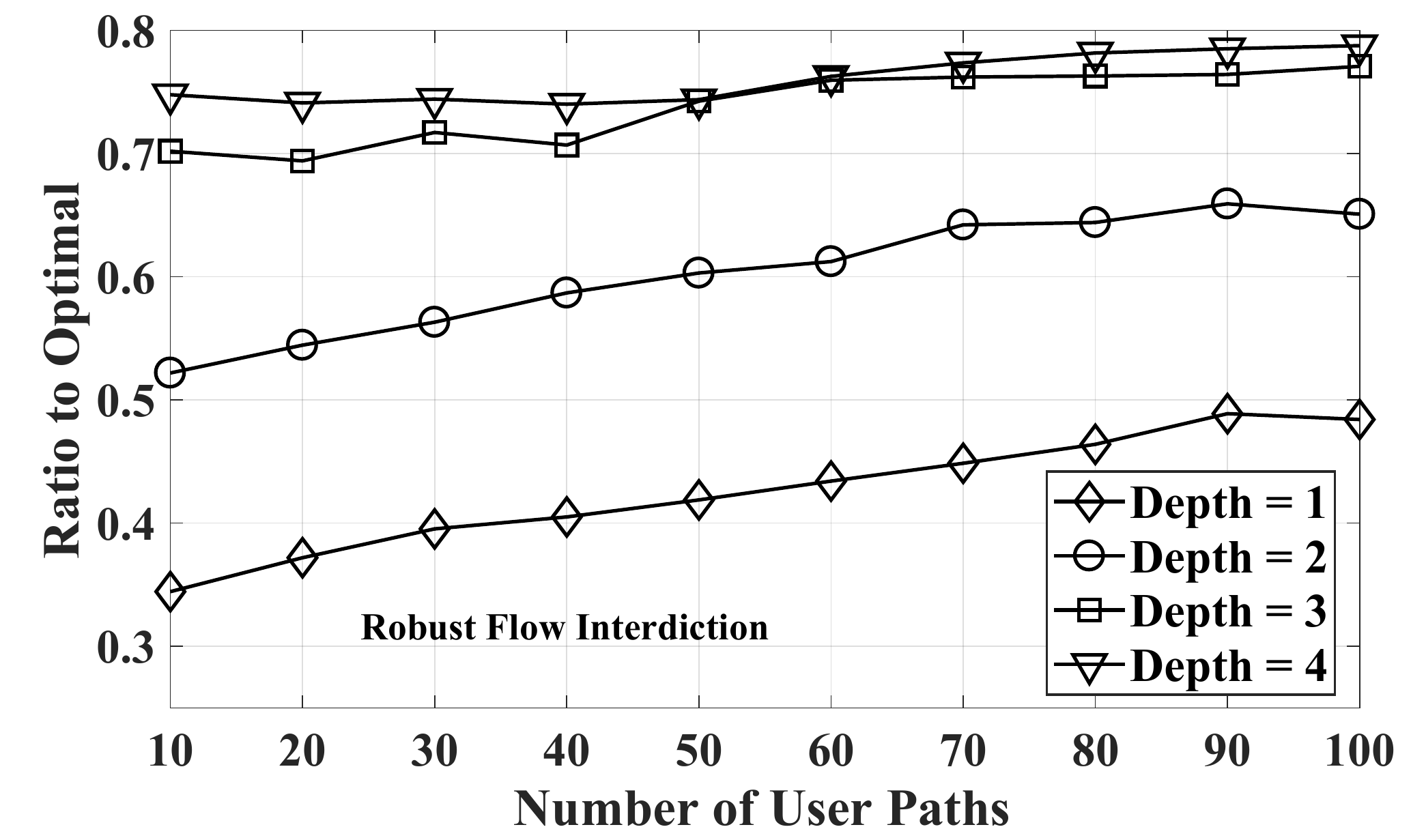}
		\vspace{-4mm}
		\label{fig:non_disjoint1000_robust}
	\end{minipage}%
	
}

\vspace{-3mm}
\caption{Ratio of throughput reduction of the solutions by our algorithms to the optimal.}
%\vspace{-5.7mm}
\label{fig:sim_figures}		
\end{figure*}

\subsection{Robust Flow Interdiction} 
In the robust flow interdiction, we randomly select 10 groups of $k$ paths as the uncertainty set $\mathcal{U}$ for $k\in\{10,20,\ldots,100\}$. Similar as before, we randomly selected 5 source-destination pairs for the interdictor in each network and form 100 scenarios for each $k$. 

\subsubsection{Algorithms Involved in Performance Comparisons}
We embed the Extended Recursive Greedy algorithm with different recursion depths in our proposed approximation framework (\textbf{Algorithm \ref{algorithm:robust}}). The optimal solution in this case is obtained by solving the LP (\ref{LP:robust}).
\subsubsection{Performance Metric}
For all strategies $w$ computed by our algorithms, we calculate the ratio of $\min_{P\in\mathcal{U}}\Lambda(w,P)$ to that of the optimal. The results reported are again averaged over all the scenarios.
\subsubsection{Simulation Results}

We plot the results in Figure \ref{fig:non_disjoint1000_robust}. Taking the depth as four, our approximation framework achieves interdiction strategies that are more than 70\% of the optimal (0.7-approximation). As in the previous case, we have implemented the framework with recursion depth of five and six but found that it did not improve the performance. 
%Comparing to the ratio achieved in the deterministic flow interdiction, we can see that the approximation loss caused by our framework is only about 0.76 on average.
%
%\begin{figure}
%	\centering
%	\includegraphics[width=0.9\linewidth,height = 0.18\textheight]{Simulations/Figures/non_disjoint_1000_robust_1}
%	\caption{Ratio of the minimum throughput reduction for all $P\in\mathcal{U}$ of the strategies computed by our approximation framework to the optimal.}
%	\label{fig:non_disjoint1000_robust}
%\end{figure}

\section{Discussion of General Networks} \label{sec:discussion}
%In this section, we provide a brief discussion on extension of relaxing the acyclicness of network and the low-rateness of interdictor.
In this section, we extend our network interdiction paradigm and the two flow interdiction problems to general networks.
Our network interdiction paradigm can be straightforwardly extended to general networks by allowing the network graph to be a general directed graph. One caveat is that we need to additionally restrict the flows that the interdictor injects to be free of cycles. Since otherwise, as the flow value of a cycle is zero, the interdictor would be able to consume the capacities of the edges in any cycle without spending any  of its budget, which would lead to meaningless solutions. Under the generalized paradigm, the deterministic and robust flow interdiction problems can be defined in the same way as Definitions \ref{def:deterministic} and \ref{def:robust}.
 For the network interdiction paradigm on general networks, Proposition \ref{proposition:property} still holds. But the Extended Recursive Greedy algorithm will break down since the edge set it returns will be an $s$-$t$ walk instead of an $s$-$t$ path (i.e. it may contain cycles). Furthermore, we can prove by an approximation-preserving reduction from the Longest Path problem in directed graphs \cite{cite:longestpath} that there is no polynomial time (quasi-polynomial time) algorithm with an approximation ratio of $O(n^{1-\delta})$ for any $\delta>0$ unless $P = NP$ ($DTIME(O(n^{\log n})) = NP$).\footnote{$DTIME(n^{\log n})$ denotes the class of problems that can be solved in quasi-polynomial time.} The reduction works by defining the graph in the Longest Path problem instance as the network graph and designating each edge as a user path. We further set the capacities of the edges and the interdictor's budget as one. Thus, the optimal single-path flow would essentially be the longest path from the interdictor's source and destination, with the throughput reduction equaling the length of the path it corresponds to. Enumerating all the node pairs in the graph, we can get the longest path in the original graph if we can solve the deterministic flow interdiction problem.
 This implies that the two flow interdiction problems on general directed graph are extremely hard to approximate within a non-trivial factor in polynomial or even quasi-polynomial time. 

%\section{Conclusion and Future Directions}  \label{sec:conclusion}
\section{Conclusion} \label{sec:conclusion}
In this paper, we proposed a new paradigm for network interdiction that models the interdictor's action as injecting bounded-value flows to maximally reduce the throughput of the residual network. We studied two problems under the paradigm: deterministic flow interdiction and robust flow interdiction, where the interdictor has certain or uncertain knowledge of the operation of network users, respectively. Having proved the computation complexity of the two problems, we proposed an algorithm with logarithmic approximation ratio and quasi-polynomial running time was proposed for the deterministic flow interdiction. We further developed an approximation framework that integrates the algorithm and forms a quasi-polynomial time procedure that approximates the robust flow interdiction within a poly-logarithmic factor. Finally, we evaluated the performance of the proposed algorithms through simulations.

%There are several future directions that are worthy of exploration. First, it is interesting to prove stronger approximation hardness results for the two flow interdiction problems on acyclic networks. Second, it would be meaningful to improve the approximation ratio and time complexity of the proposed algorithms. The first two directions combined help close the gap between the algorithmically achievable upper bound and lower bound of the flow interdiction problems. Finally, incorporating interdictors that are not low-rate may bring about different results and expand the applicability of the paradigm.

\begin{appendices}
	\section{Proof of Proposition {\ref{proposition:property}}} \label{appendix:proofproperty}
			Let $w$ be an interdiction strategy. If it is a distribution on single-path flows, then the proposition follows. Otherwise, there exists an $\mathbf{f}^*$ with $w(\mathbf{f}^*)>0$ that is not a single-path flow. By the flow decomposition theorem \cite{cite:networkflow} and that the network is acyclic, we can decompose $\mathbf{f}^*$ into $\mathbf{f}^*=\sum_i\mathbf{q}_i$, where $\mathbf{q}_1,\ldots,\mathbf{q}_r$ are single-path flows from $s$ to $t$. We further define $\mathbf{q}_i^*=\frac{\gamma}{val(\mathbf{q}_i)}\mathbf{q}_i$ as a scaled version of $\mathbf{q}_i$ with value $\gamma$, for $i\in\{1,\ldots,r\}$. Note that $\sum_i^rval(\mathbf{q}_i)=\gamma$, and since $\gamma\le\min_eC(e)$, $\mathbf{q}_1^*,\ldots,\mathbf{q}_r^*$ are all valid single-path flows in $\mathcal{F}_{\gamma}$. In the following, we show that we can redistribute the probability that $w$ lays on $\mathbf{f}^*$ to all its component single-path flows by decreasing $w(\mathbf{f}^*)$ to zero and adding $\frac{val(\mathbf{q}_i)}{\gamma}w(\mathbf{f}^*)$ to each $w(\mathbf{q}_i)$, with the resulting interdiction strategy $w'$ satisfying $\Lambda(w',P)\ge \Lambda(w,P)$. Repeating the process for all non-single-path flows $\mathbf{f}^*$ with $w(\mathbf{f}^*)>0$, we prove the proposition.
		
		For any $P$, we write the linear program (\ref{throughput}) with respect to $\mathbf{f^*}$ and $P$ in vector form and construct its dual as follows:
		\begin{align*}
		\quad\text{maximize } & \mathbf{1}^{\top}\bm{\tilde{\lambda}} & \mbox{minimize }& {\tilde{C}}^{\top}_{\mathbf{f}^*}\mathbf{g}_0+\bm{\lambda}^\top\mathbf{g}_1\\
		\text{\textbf{s.t. }} \mathbf{A}\bm{\tilde{\lambda}}&\le {\tilde{C}_{\mathbf{f}^*}} & 	\mbox{\textbf{s.t. }} &	\mathbf{A}^\top\mathbf{g}_0+\mathbf{Ig}_1 \ge \mathbf{1}	 \\
		\mathbf{0}\le&\bm{\tilde{\lambda}}\le\bm{\lambda} 
		& &\quad\mathbf{g}_0,\mathbf{g}_1\ge \mathbf{0}\\				
		\end{align*} 
		where $\bm{\lambda}=(\lambda_1,\ldots,\lambda_k)^\top$ is the vector of the initial flow values, $\tilde{C}_{\mathbf{f^*}}=C-\mathbf{f^*}$ , $\mathbf{A}$ is the matrix representation of constraints (\ref{capacity}), $\mathbf{I},\mathbf{1},\mathbf{0}$ are the identity matrix, the all-1 vector and the all-0 vector, and $\mathbf{g}_0,\mathbf{g}_1$ are the dual variables. By the strong duality theorem \cite{cite:linearprogramming}, the optimal value $T(\mathbf{f^*},P)$ of the primal problem is equal to $\tilde{C}_{\mathbf{f^*}}^\top \mathbf{g}^*_0+\bm{\lambda}^\top\mathbf{g}^*_1$, where $\mathbf{g}_0^*,\mathbf{g}_1^*$ is an optimal basic feasible solution to the dual problem.  Furthermore, consider the linear program (\ref{throughput}) with respect to each $\mathbf{q}_i^*$ and its dual. Note that the dual has the same feasible region as that associated with $\mathbf{f}^*$. Therefore, $\mathbf{g}_0^*,\mathbf{g}_1^*$ is still a basic feasible solution. Now, invoking weak duality, we have
		\[\forall i,\quad T(\mathbf{q}_i^*,P)\le (C-\mathbf{q}_i^*)^\top\mathbf{g}_0^*+\bm{\lambda}^\top \mathbf{g}_1^*.\] 
		It follows that 
		\begin{align*}
		&\sum_i\frac{val(\mathbf{q}_i)w(\mathbf{f}^*)}{\gamma}T(\mathbf{q}_i^*,P)\\&\le \sum_i\frac{val(\mathbf{q}_i)w(\mathbf{f}^*)}{\gamma}\left((C-\mathbf{q}_i^*)^\top\mathbf{g}_0^*+\bm{\lambda}^\top \mathbf{g}_1^*\right)\\
		&=w(\mathbf{f}^*)[(C-\mathbf{f}^*)^\top\mathbf{g}_0^*+\bm{\lambda}^\top \mathbf{g}_1^*]=w(\mathbf{f}^*)T(\mathbf{f}^*,P).
		\end{align*}
		Hence, we have $\Lambda(w',P)\ge\Lambda(w,P)$, and the proposition follows.

	\section{Proof of Proposition {\ref{proposition:hardness}}} \label{appendix:proofhardness}
	The proof is done by reduction from the 3-satisfiability problem, which is a classical NP-Complete problem \cite{cite:karp}.
	\textit{3-satisfiability: Given a set of boolean variables $x_i,1\le i\le n$ and a formula $C_1\vee C_2\vee\ldots\vee C_k$ with each clause $C_j$ being a disjunction ($\wedge$) of at most three literals $x_i$ or \ov{$x_i$}, the 3-satisfiability asks whether there is a satisfying assignment, i.e., an assignment of the variables that makes the formula true.}
	
	Given an instance of 3-satisfiability, the corresponding instance of the deterministic flow interdiction is constructed as follows. To begin with, without loss of generality, we assume that there is no clause that contains both $x_i$ and \ov{${x_i}$} for some $i$, as such clause can be satisfied by all the assignments. To create the network, we first add a path $p_j$ with $3n$ edges for each clause $C_j$. The paths are node-disjoint. Then for each variable $x_j$, we create a variable gadget with three nodes $u_i,v_{i0},v_{i1}$ and two edges $(u_i,v_{i0}),(u_i,v_{i1})$. Nodes $v_{i0},v_{i1}$ correspond to \ov{${x_i}$}, $x_i$, respectively. We next connect the variable gadgets and the paths for the clauses. For each $v_{i_0}$ ($v_{i_1}$), let $C_{i1},\ldots,C_{ir}$ be the set of clauses that contains literal \ov{${x_i}$}$(x_i)$. Let $e_1,\ldots,e_r$ be the $3i$-th ($3i+1$-th) edges on $p_{i1},\ldots,p_{ir}$. We add edges to the network to sequentially connect $v_{i_0},e_1,\ldots,e_r,u_{i+1}$ and refer to the resulting path from $v_{i_0}$ to $u_{i+1}$ as $v_{i_0}$-$u_{i+1}$ segment. For $i=n$, we further add a node to serve as $u_{n+1}$. We designate $s=u_0$ and $t=u_{n+1}$ as the source and the destination of the interdictor. The set of user paths is $P=\{p_1,\ldots, p_k\}$ and the initial flow values $f_1=\ldots=f_k=1$. The capacities of all the edges, and the budget of the interdictor are set to 1. Now we have completed the construction of the corresponding instance of the deterministic flow interdiciton. Note that the constructed network is a DAG and the whole reduction process can be done in polynomial time. See Figure \ref{fig:reduction} for an illustration of the reduction process.
	
	We proceed to show that there exists a single-path flow $\mathbf{f}$ with $\Lambda(\mathbf{f},P)=k$ if and only if there is a satisfying assignment for the original 3-satisfiability instance. First, if there exists a satisfying assignment with $x_i=a_i\in\{0,1\}$, we claim that the single-path flow $\mathbf{f}$ that corresponds to the $s$-$t$ path consisting of $(u_i,v_{ia_i})$ and $v_{ia_i}$-$u_{i+1}$ segment for all $i$ has throughput reduction $k$. Since in the satisfying assignment, each clause is set true by some literal, we have that each user path contains an edge with zero residual capacity after interdicted by $\mathbf{f}$. It follows that $T(\mathbf{f},P)=0$ and $\Lambda(\mathbf{f},P)=k$. Second, if there exists a single-path flow $\mathbf{f}$ with $\Lambda(\mathbf{f},P)=k$, then the path $E_{\mathbf{f}}$ that $\mathbf{f}$ corresponds to must intersect with all user paths. We next show that $E_{\mathbf{f}}$ can be converted to an $s$-$t$ path $E_{\mathbf{f}'}$ consisting only of $(u_i,v_{ia_i})$'s and $v_{ia_i}$-$u_{i+1}$ segments and also intersect with all user paths. Indeed, $E_{\mathbf{f}'}$ can be constructed by taking all the $(u_i,v_{ia_i})$'s and $v_{ia_i}$-$u_{i+1}$ segments that $E_\mathbf{f}$ intersects. Note that since there is no clause that contains both $x_i$ and \ov{${x_i}$} for some $i$, the assignment $\forall i,\ x_i=a_i$ induced by $E_{\mathbf{f}'}$ is a valid assignment. Since $E_{\mathbf{f}'}$ intersects with all user paths, the assignment satisfies all the clauses, and thus makes the formula true. Hence, we justify the validity of the reduction. Combining with corollary 1, we have that the deterministic flow interdiction problem is NP-hard.
	
	\begin{figure}
		\centering
		\includegraphics[width=0.9\linewidth]{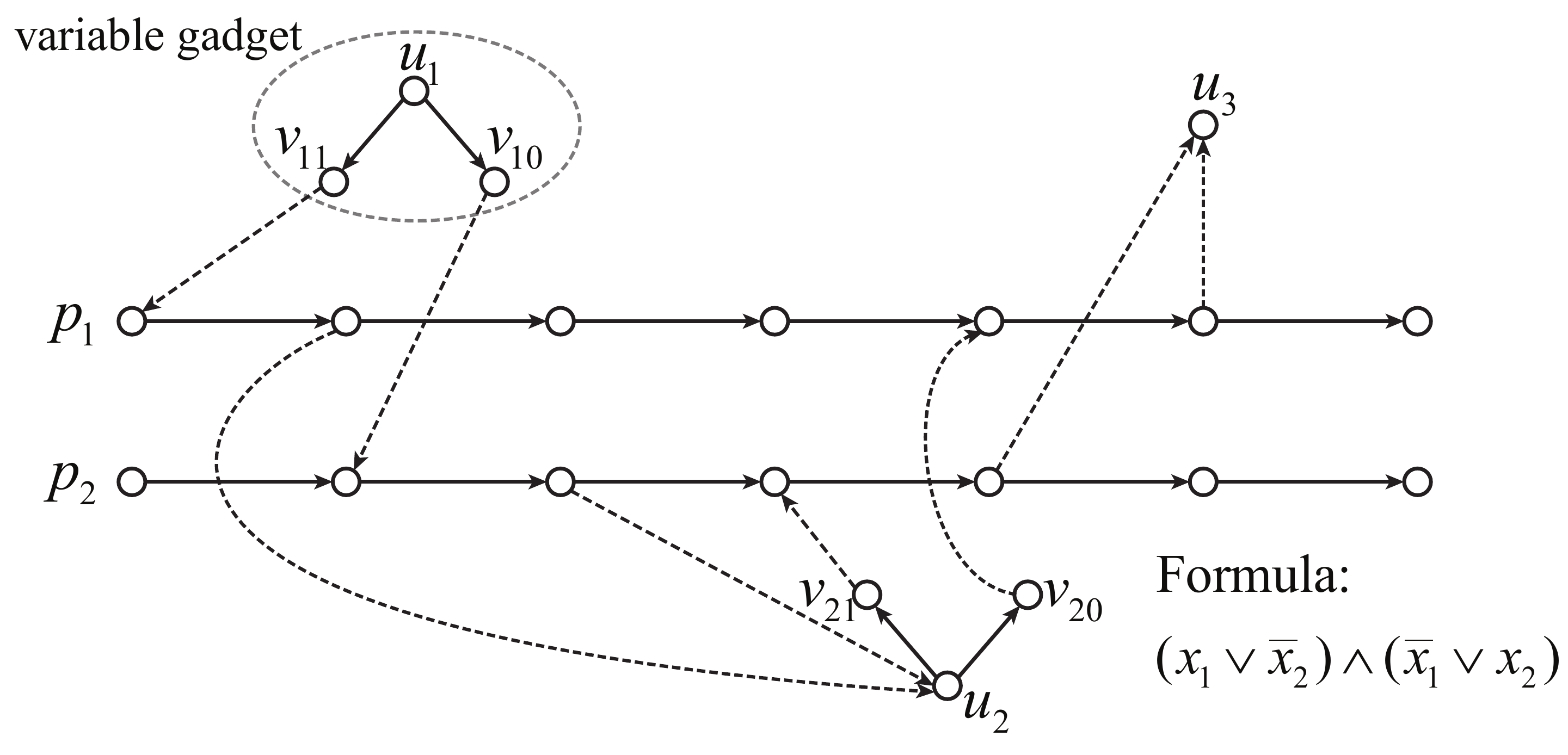}
		\caption{An illustration of the reduction process for the 3-satisfiability instance with formula $(x_1\vee$ \protect\ov{$x_2$}$)\wedge ($\protect\ov{${x_1}$}$\vee x_2)$.}
		\label{fig:reduction}
	\end{figure}

\section{Generalization of the Recursive Greedy Algorithm} \label{app:generalization}
In this section, we describe a generalized version of the Recursive Greedy algorithm that uses more than one anchors. Let $a>1$ be some integers. The details of the algorithm is presented in \textbf{Algorithm \ref{algorithm:greedy1}}. At step 8, instead of going over all $v\in V$, the generalized algorithm goes through all $a-1$ combinations of nodes in $V$ and uses them as anchors. The analysis of the algorithm is given in Theorem \ref{thm:recgreedy1}
\begin{algorithm}
	\caption{The Generalized Recursive Greedy Algorithm}
	\begin{algorithmic}[1]	\label{algorithm:greedy1}		
		\REQUIRE{Network graph $G(V,E)$,  user paths $P=\{p_1,\ldots,p_2\}$ with initial flow values $\{f_1,\ldots,f_k\}$, Interdictor's source $s$, destination $t$ and budget $\gamma$ }\\
		\ENSURE{The optimal $s$-$t$ path $E_\mathbf{f}$}
		
		%			\Procedure {DesignTranScheme}{$\mathcal{T}$}
		%			\EndProcedure
		%Precompute the elements in the corresponding Finite Horizon Markov Decision Process.\\
		\STATE \textbf{Run: } $RG(s,t,\emptyset,I)$\\
		\textit{The Recursive Function $RG(u_1,u_2,X,i)$:}
		\STATE $E_\mathbf{f}:=$ shortest $s$-$t$ path.
		\IF{$E_\mathbf{f}$ does not exist}
		\RETURN Infeasible
		\ENDIF
		\IF{$i=0$}
		\RETURN $E_\mathbf{f}$
		\ENDIF
		\STATE $r:=\Lambda_X(E_\mathbf{f},P)$.
		\FOR{$v_1,v_2,\ldots,v_{a-1}\in V$}
		\STATE $E_{\mathbf{f}_1}:=RG(u_1,v_1,X,i-1),E_{\mathbf{f}_2}:=RG(v_1,v_2,X\cup E_{\mathbf{f}_1},i-1),\ldots,E_{\mathbf{f}_{a}}:=RG(v_{a-1},u_2,X\cup E_{\mathbf{f}_1}\cup\ldots\cup E_{\mathbf{f}_{a-1}},i-1)$.
		\IF{$\Lambda_X(E_{\mathbf{f_1}}\cup\ldots\cup E_{\mathbf{f_{a}}},P)>r$}
		\STATE $r:=\Lambda_X(E_{\mathbf{f_1}}\cup\ldots\cup E_{\mathbf{f_a}},P)$, $E_\mathbf{f}:=E_{\mathbf{f}_1}\cup\ldots\cup E_{\mathbf{f}_a}$.
		\ENDIF		
		\ENDFOR
		\RETURN{$E_\mathbf{f}$}
	\end{algorithmic}
\end{algorithm}

\begin{theorem}\label{thm:recgreedy1}
	If $I\ge \lceil\log_{a}d\rceil$, the Generalized Recursive Greedy algorithm returns an $s$-$t$ path $E_\mathbf{f}$ with $\Lambda(E_\mathbf{f},P)\ge \frac{1}{\lceil\log_{a}d\rceil+1}\Lambda(E_\mathbf{f^*},P)$, where $d$ is the length of $E_{\mathbf{f}^*}$.
\end{theorem}

%%%Introduction, and this proof!
\begin{proof}
	We prove a more general claim, that for all $u_1,u_2\in V,\ X\subseteq E$,  if $I\ge \lceil\log d\rceil$, the procedure $RG(u_1,u_2,X,I)$ returns an $u_1$-$u_2$ path $E_\mathbf{f}$ with $\Lambda_X(E_\mathbf{f},P)\ge \frac{1}{\lceil\log d\rceil+1}\Lambda_X(E_\mathbf{f^*},P)$, where $E_{\mathbf{f}^*}$ is the $u_1$-$u_2$ path that maximized $\Lambda(\cdot,P)$ and $d$ is the length of $E_{\mathbf{f}^*}$. The theorem follows from the claim by setting $u_1=s$, $u_2=t$ and $X=\emptyset$.
	
	Let the nodes on the path $E_\mathbf{f^*}$ be $\{u_1=v_0,\ldots,v_d=u_2\}$. The proof is done by induction on $d$.	
	First, for the base step, when $d=1$, it means that there exists an edge between $u_1$ and $u_2$, which must be the shortest $u_1$-$u_2$ path. Obviously the procedure checks this path, and the claim follows. Next, suppose the claim holds for $d\le l\in \mathbb{N}$. 
	 When $d=l+1$, let $v_1^*=v_{\lceil\frac{d}{a}\rceil}, v_2^*=v_{\lceil\frac{2d}{a}\rceil},\ldots,v_{a-1}^*=v_{\lceil\frac{(a-1)d}{a}\rceil}$. Let $E_{\mathbf{f}_1^*},\ldots,E_{\mathbf{f}_a^*}$ be the subpaths of $E_{\mathbf{f}^*}$ from $s$ to $v_1^*,\ldots,v_{a-1}^*$ to $t$. When $RG$ examines $\{v_1^*,\ldots,v_{a-1}^*\}$ at step 8, it invokes $a$ sub-procedures denoted as $RG(u_1,v_1^*,X^{(0)},I-1),\ldots,RG(v_{a-1}^*,u_2,X^{(a-1)},I-1)$. In the above notations, we use $E_{\mathbf{f}_j}$ to denote the sub-path returned by the $j$th subprocedure, $X^{(j)}$ to denote $X\cup E_{\mathbf{f}_1}\cup\ldots\cup E_{\mathbf{f}_{j}}$ for $j\in\{1,\ldots,a\}$ and $X^{(0)}=X$. Let $E'_\mathbf{f}=E_{\mathbf{f}_1}\cup\ldots\cup E_{\mathbf{f}_a}$. Our goal is to show that 
	\begin{align}
	\Lambda_X(E'_\mathbf{f},P)\ge\frac{1}{\lceil\log_{a}d\rceil+1}\Lambda_X\left(E_{\mathbf{f}^*},P\right), \label{ieq:target1}
	\end{align}
	which proves the induction step since the path $E_{\mathbf{f}}$ that $RG(s,t,X,I)$ returns must satisfy $\Lambda_X(E_{\mathbf{f}},P)\ge \Lambda_X(E'_{\mathbf{f}},P)$.
	
	Since $I\ge \lceil\log_{a}d\rceil$, we have $I-1\ge\lceil\log_{a}d\rceil-1 = \lceil\log_{a}d/a\rceil$. By the induction hypothesis,
	\begin{align*}
	\Lambda_{X^{(0)}}(E_{\mathbf{f}_1},P)&\ge \frac{1}{\lceil\log_{a}d\rceil}\Lambda_{X^{(0)}}(E_\mathbf{f_1^*},P),\\
	\Lambda_{X^{(1)}}(E_{\mathbf{f}_2},P)&\ge \frac{1}{\lceil\log_{a}d\rceil}\Lambda_{X^{(1)}}(E_\mathbf{f_2^*},P),\\
	&\ldots,\\\Lambda_{X^{(a-1)}}(E_{\mathbf{f}_a},P)&\ge \frac{1}{\lceil\log_{a}d\rceil}\Lambda_{X^{(a-1)}}(E_\mathbf{f_a^*},P).
	\end{align*}
	For $j\in\{0,\ldots,a-1\}$, by the submodularity of $\Lambda$ (Lemma \ref{lemma:submodular}), we have $\Lambda_{X^{(j)}}(E_\mathbf{f_a^*},P)\ge \Lambda_{X^{(a)}}(E_\mathbf{f_a^*},P)$. Using this, we sum all the inequalities above and get
	\begin{align}
	\Lambda_X(E'_\mathbf{f},P)&=\sum_{j=0}^{a-1} \Lambda_{X^{(j)}}(E_{\mathbf{f}_{j+1}},P)\\
	&\ge \frac{1}{\lceil\log_{a}d\rceil}\sum_{j=0}^{a-1}\Lambda_{X^{(j)}}(E_{\mathbf{f}_{{j+1}}^*},P)\\
	&\ge \frac{1}{\lceil\log_{a}d\rceil}\sum_{j=0}^{a-1}\Lambda_{X^{(a)}}(E_\mathbf{f_{j+1}^*},P).
	\end{align}
	Again, by Lemma \ref{lemma:submodular}, we have for $j\in\{0,\ldots,a-1\}$,
	\[
	\Lambda_{X^{(a)}}(E_{\mathbf{f}_{j+1}^*},P)\ge \Lambda_{X^{(a)}\cup\left(\bigcup_{i=1}^{j}E_{\mathbf{f}_i^*}  \right)}(E_{\mathbf{f}_{j+1}^*},P)
	\]
	It follows that 
	\begin{small}	
		\begin{align}
		\Lambda_X(E'_\mathbf{f},P)&\ge \frac{1}{\lceil\log_{a}d\rceil}\sum_{j=0}^{a-1}\Lambda_{X^{(a)}\cup\left(\bigcup_{i=1}^{j}E_{\mathbf{f}_i^*}  \right)}(E_{\mathbf{f}_{j+1}^*},P)\\
		&=\frac{1}{\lceil\log_{a}d\rceil}\Lambda_{X^{(a)}}(E_{\mathbf{f}_1^*}\cup\ldots\cup E_{\mathbf{f}_a^*},P)\label{eq:def11}\\
		&=\frac{1}{\lceil\log_{a}d\rceil}\left[\Lambda\left(X\cup\left(\bigcup_{j=1}^aE_{\mathbf{f_j}}\right)\cup\left(\bigcup_{j=1}^aE_{\mathbf{f_j^*}}\right),P\right)\right.\nonumber\\&\left.\quad-\Lambda\left(X\cup\left(\bigcup_{j=1}^aE_{\mathbf{f_j}}\right),P\right)\right]\label{eq:def21}\\
		&=\frac{1}{\lceil\log_{a}d\rceil}\left[\Lambda_X\left(\left(\bigcup_{j=1}^aE_{\mathbf{f_j}}\right)\cup\left(\bigcup_{j=1}^aE_{\mathbf{f_j^*}}\right),P\right)\right.\nonumber\\&\left.\quad-\Lambda_X\left(\bigcup_{j=1}^aE_{\mathbf{f_j}},P\right)\right]\label{eq:def31}\\
		&\ge\frac{1}{\lceil\log_{a}d\rceil}\left[\Lambda_X\left(\bigcup_{j=1}^aE_{\mathbf{f_j^*}},P\right)-\Lambda_X\left(\bigcup_{j=1}^aE_{\mathbf{f_j}},P\right)\right]\label{ieq:monotonicity1}\\
		&= \frac{1}{\lceil\log_{a}d\rceil}\left[\Lambda_X\left(E_{\mathbf{f^*}},P\right)-\Lambda_X\left(E'_{\mathbf{f}},P\right)\right]\label{ieq:finalbound11}, 
		\end{align}
	\end{small}
	where inequality (\ref{ieq:monotonicity1}) follows from the monotonicity of $\Lambda$ and equalities (\ref{eq:def11}), (\ref{eq:def21}) and (\ref{eq:def31})  follow from the definition of $\Lambda_X$.  From (\ref{ieq:finalbound11}), we obtain (\ref{ieq:target1}),
	which concludes the proof.
\end{proof}

%%Better analysis of the time complexity
\textbf{Time Complexity: } As we invoke at most $an^a$ sub-procedures at each level of recursion and the computation of $\Lambda$ takes $O(m)$ time, the time complexity of the Generalized Recursive Greedy algorithm is $O((an)^{(a-1)I}m)$. Again, taking $I=\log_a n$, we get an $1/(\lceil\log_an\rceil+1)$-approximation with a time complexity of $O((an)^{(a-1)\log_an}m)$.

\section{Proof of Lemmas \ref{lemma:extendproperties} and \ref{lemma:gap}}\label{app:prooflemma}
This section is devoted to the proof of Lemmas \ref{lemma:extendproperties} and \ref{lemma:gap}. We define 
$\prod\limits_{e\in p_i,e\in E_2,\tilde{C}_A(e)\le \sum_{p_j\ni e}\lambda_j}\frac{\tilde{C}_A(e)}{\sum_{p_j\ni e}\lambda_j}$ as $\Delta_i(A,P)$.

\subsection{Lemma \ref{lemma:extendproperties}} \label{app:prooflemma2}
	Recall the definition of submodularity and monotonicity in Lemma \ref{lemma:submodular}. First, we can easily see from {Phase I} and {Phase II} that for all $i$, $\tilde{\lambda}_{iA}^{(2)}$ is monotonically non-increasing with respect to $A$. It follows that $\bar{\Lambda}(\cdot,P)$ is monotone.

Next, we prove the submodularity of $\bar{\Lambda}$. Consider two sets $A\subseteq B\subseteq E$ and an edge $e\in E,e\notin B$. Our goal is to show that $\bar{\Lambda}(A\cup\{e\},P)-\bar{\Lambda}(A,P)\ge \bar{\Lambda}(B\cup\{e\},P)-\bar{\Lambda}(B,P)$. We divide the proof into three cases. 

\textit{Case I:} If $e\notin E_0$, then $\bar{\Lambda}(A\cup\{e\},P)-\bar{\Lambda}(A,P)=\bar{\Lambda}(B\cup\{e\},P)-\bar{\Lambda}(B,P)=0$. 

\textit{Case II:} If $e\in E_1$, then suppose $e\in p_i$ for some $i$. 
%Denote $\tilde{f}_i^{(1)}, \tilde{f}_i^{(2)}$ in the procedure for $\bar{\Lambda}(X,P)$ as $\tilde{f}_{iX}^{(1)},\tilde{f}_{iX}^{(2)}$, $X=A,A\cup\{e\},B,B\cup\{e\}$. 
Note that since $A\subseteq B$, $\tilde{\lambda}_{iA}^{(1)}\ge \tilde{\lambda}_{iB}^{(1)}$. We further divide this case into three subcases. (i). If $C(e)-\gamma\ge\tilde{\lambda}_{iA}^{(1)} $, then we have $\tilde{\lambda}_{iA}^{(1)}=\tilde{\lambda}_{iA\cup\{e\}}^{(1)}$ and $\tilde{\lambda}_{iB}^{(1)}=\tilde{\lambda}_{iB\cup\{e\}}^{(1)}$. Hence,
\[\bar{\Lambda}(A\cup\{e\},P)-\bar{\Lambda}(A,P)=\bar{\Lambda}(B\cup\{e\},P)-\bar{\Lambda}(B,P)=0.\]
(ii). If $\tilde{\lambda}_{iA}^{(1)}> {C}(e)-\gamma\ge \tilde{\lambda}_{iB}^{(1)}$, then $\tilde{\lambda}_{iA}^{(1)}>\tilde{\lambda}_{iA\cup\{e\}}^{(1)}$ and $\tilde{\lambda}_{iB}^{(1)}=\tilde{\lambda}_{iB\cup\{e\}}^{(1)}$. Hence,
\begin{align*}
\bar{\Lambda}(A\cup\{e\},P)-\bar{\Lambda}(A,P)>0,\\ \bar{\Lambda}(B\cup\{e\},P)-\bar{\Lambda}(B,P)=0. 
\end{align*}
(iii). If $C(e)-\gamma< \tilde{\lambda}_{iB}^{(1)}$, we have $\tilde{\lambda}_{iA\cup\{e\}}^{(1)}=\tilde{\lambda}_{iB\cup\{e\}}^{(1)}=C(e)-\gamma$. It follows that
\begin{align*} 
\tilde{\lambda}_{iA\cup\{e\}}^{(2)}=(C(e)-\gamma)\cdot\Delta_i(A\cup\{e\},P),\\
\tilde{\lambda}_{iB\cup\{e\}}^{(2)}=(C(e)-\gamma)\cdot\Delta_i(B\cup\{e\},P).
\end{align*}
As $A\cup\{e\}\subseteq B\cup\{e\}$, we have $\Delta_i(A\cup\{e\},P)\ge \Delta_i(B\cup\{e\},P)$. It follows that, 
\begin{align*}
\bar{\Lambda}(A\cup\{e\},P)-\bar{\Lambda}(A,P)&=(\tilde{\lambda}_{iA}^{(1)}-\tilde{\lambda}_{iA\cup\{e\}}^{(1)})\Delta_i(A\cup\{e\},P)\\
&\ge (\tilde{\lambda}_{iB}^{(1)}-\tilde{\lambda}_{iB\cup\{e\}}^{(1)})\Delta_i(B\cup\{e\},P)\\
&= \bar{\Lambda}(B\cup\{e\},P)-\bar{\Lambda}(B,P),
\end{align*} 
where the two equalities follow from the fact that $\Delta_i(A\cup\{e\},P)=\Delta_i(A,P), \Delta_i(B\cup\{e\},P)=\Delta_i(B,P)$ since $e\notin E_2$. 

\textit{Case III:} If $e\in E_2$, then when $C(e)-\gamma\ge \sum_{p_j\ni e}\lambda_j$, we still have  \[\bar{\Lambda}(A\cup\{e\},P)-\bar{\Lambda}(A,P)=\bar{\Lambda}(B\cup\{e\},P)-\bar{\Lambda}(B,P)=0.\] 
%When $C(e)-\gamma< \sum_{e\in p_j}f_j$, let $I_e\subseteq\{1,\ldots,k\}$ be the set of indices of the user paths that $e$ is on. 
When $C(e)-\gamma< \sum_{p_j\ni e}\lambda_j$, first, we observe that since $e\notin E_1$, $\tilde{\lambda}_{iA}^{(1)}=\tilde{\lambda}_{iA\cup\{e\}}^{(1)}$ and $\tilde{\lambda}_{iB}^{(1)}=\tilde{\lambda}_{iB\cup\{e\}}^{(1)}$. Also, as $A\subseteq B$, we have $\tilde{\lambda}_{iA}^{(1)}\ge \tilde{\lambda}_{iB}^{(1)}$ and $\Delta_i(A,P)\ge\Delta_i(B,P)$ for all $i$. Combining these, we obtain
%	\begin{small}
\begin{align*}
&\bar{\Lambda}(A\cup\{e\},P)-\bar{\Lambda}(A,P)\\&=\sum_{i:e\in p_i}\tilde{\lambda}_{iA}^{(1)}\cdot\Delta_i(A,P)\cdot\left(1-\frac{{C}(e)-\gamma}{\sum_{p_j\ni e}\lambda_j}\right)\\
&\ge \sum_{i:e\in p_i}\tilde{\lambda}_{iB}^{(1)}\cdot\Delta_i(B,P)\cdot\left(1-\frac{{C}(e)-\gamma}{\sum_{p_j\in e}\lambda_j}\right)\\
&= \bar{\Lambda}(B\cup\{e\},P)-\bar{\Lambda}(B,P),
\end{align*}
%	\end{small}

%\begin{small}
%	\begin{align*}
%	\bar{\Lambda}(A\cup\{e\},P)-\bar{\Lambda}(A,P)&=\sum_{i\in I_e}f_{iA}^{(1)}\cdot\Delta_i(A,P)\cdot\left(1-\frac{\tilde{C}(e)}{\sum_{e\in p_j}f_j}\right)\\
%	&\ge \sum_{i\in I_e}f_{iB}^{(1)}\cdot\Delta_i(B,P)\cdot\left(1-\frac{\tilde{C}(e)}{\sum_{e\in p_j}f_j}\right)\\
%	&= \bar{\Lambda}(B\cup\{e\},P)-\bar{\Lambda}(B,P),
%	\end{align*}
%\end{small}
Therefore, in all cases, we have $\bar{\Lambda}(A\cup\{e\},P)-\bar{\Lambda}(A,P)\ge \bar{\Lambda}(B\cup\{e\},P)-\bar{\Lambda}(B,P)$. Hence, $\bar{\Lambda}(\cdot,P)$ is submodular.

\subsection{Lemma \ref{lemma:gap}} \label{app:prooflemma3}
In the definition of $\bar{\Lambda}$, we have reasoned that $\Lambda(A,P)\le \bar{\Lambda}(A,P)$. What is left is to show that $\bar{\Lambda}(A,P)\le (b+1)\cdot\Lambda(A,P)$. Let $\tilde{\lambda}_1,\ldots,\tilde{\lambda}_k$ be an optimal solution to the maximization problem (\ref{throughput}) associated with $A$.
%In the following, all the $\tilde{f}_i^{(1)}$'s, $\tilde{f}_i^{(2)}$'s and $\tilde{C}$ are associated with $A$.

First, for Phase I, observe that for all $i$, $\tilde{\lambda}_i\le \min\{\lambda_i,\min\limits_{e\in p_i,e\in E_1}\{\tilde{C}_A(e)\}\}$. Therefore, we have $\sum_i\tilde{\lambda}_{iA}^{(1)}\ge T(A,P)$. It follows that
\begin{align}
\Lambda(A,P)\ge \sum_i(\lambda_i-\tilde{\lambda}_{iA}^{(1)}). \label{ieq:phaseI}
\end{align}
Next, for phase II,
\begin{small} 
	\begin{align}
	&\sum_i\left(\tilde{\lambda}_{iA}^{(1)}-\tilde{\lambda}_{iA}^{(2)}\right) = \sum_i \tilde{\lambda}_{iA}^{(1)}\left(1-\Delta_i(A,P)\right)\nonumber\\
	\le& \sum_i \lambda_{i}\left(1-\Delta_i(A,P)\right) \nonumber\\
	\le& \sum_i\lambda_i\cdot\sum_{e\in p_i,e\in E_2,\tilde{C}_A(e)\le \sum_{p_j\ni e}\lambda_j}\left(1-\frac{\tilde{C}_A(e)}{\sum_{p_j\ni e}\lambda_j}\right)\label{eq:auxiliaryieq}\\
	=&\sum_{e\in E_2,\tilde{C}_A(e)\le \sum_{p_j\ni e}\lambda_j}\left(\sum_{p_j\ni e}\lambda_j-\tilde{C}_A(e)\right)\label{eq:rearrange},
	\end{align}
\end{small}
where inequality (\ref{eq:auxiliaryieq}) follows from the fact that $1-\prod_{j}a_j\le \sum(1-a_j)$ for $0\le a_1\le\ldots\le a_j\le 1$ and equality (\ref{eq:rearrange}) comes from rearranging the terms. Since $\sum_{p_j\ni e}\tilde{\lambda}_j\le \tilde{C}_A(e)$ for all $e$, we have
\begin{align}
&\sum_{e\in E_2,\tilde{C}_A(e)\le \sum_{p_j\ni e}\lambda_j}\left(\sum_{p_j\ni e}\lambda_j - \tilde{C}_A(e) \right) \nonumber\\ &\le \sum_{e\in E_2,\tilde{C}_A(e)\le \sum_{e\in p_j}\lambda_j}\left(\sum_{p_j\ni e}\lambda_j - \sum_{p_j\ni e}\tilde{\lambda}_j \right)\nonumber\\
&\le b\cdot \sum_i\left(\lambda_i-\tilde{\lambda}_i\right)=b\cdot \Lambda(A,P). \label{ieq:phaseII}
\end{align}
Therefore, combining (\ref{ieq:phaseI}) and (\ref{ieq:phaseII}), we have $\bar{\Lambda}(A,P)=\sum_i\left(\lambda_i-\tilde{\lambda}_{iA}^{(1)}+\tilde{\lambda}_{iA}^{(1)}-\tilde{\lambda}_{iA}^{(2)}\right)\le (b+1)\Lambda(A,P)$.

\section{Justification of Integrality Assumption of $\Lambda$} \label{app:integrality}
In this section, we show that not much generality is lost if we consider the throughput reduction function $\Lambda$ to take value in integers that are bounded by some polynomial of $n$. Specifically, we show the following proposition. 
\begin{proposition}\label{proposition:appendix}
	For any instance $\mathcal{I}$ of the robust flow interdiction problem with throughput reduction function $\Lambda$ and some fixed $\epsilon>0$, we can construct another instance $\mathcal{I}'$ in whose throughput reduction function $'$ take integral values that are bounded by some polynomial of $n$ for all $w,P\in\mathcal{U}$. Furthermore, if we apply \textbf{Algorithm \ref{algorithm:robust}} to $\mathcal{I'}$, it returns a strategy $w$ that satisfies
	\begin{small}
		\begin{align*}
		&\min_{P\in\mathcal{U}}\Lambda(w,P)\\&\ge (1-\epsilon)^2 \left(\frac{N_0}{(N_0+1)(b+1)\log \frac{M}{\epsilon}\cdot(\lceil\log d\rceil+1)}\right)\min_{P\in\mathcal{U}}\Lambda(w^*,P)\end{align*}
	\end{small}
	for the original instance $\mathcal{I}$.
\end{proposition}
\begin{proof}
	First, we assume that without loss of generality, if $\Lambda(w,P)>0$ for some $w,P$, then $\Lambda(w,P)>1$, as we can always scale up all the parameters if the condition is not satisfied. With this condition, we define $\Lambda'(\mathbf{f},P)=\lfloor N_1\Lambda(\mathbf{f},P)\rfloor$, where $N$ is some integer bounded by some polynomial of $n$ and satisfies $N_1=\lceil\frac{2(b+1)\log (N_1M)(\lfloor \log d\rfloor +1)}{\epsilon}\rceil$. 
	Keeping all other parameters unchanged and substituting $\Lambda$ with $\Lambda'$, we get the new instance $\mathcal{I}'$. Note that $\Lambda'(w,P)$ satisfies the condition in the statement of the proposition. Thus, we can apply the proposed framework \textbf{Algorithm \ref{algorithm:robust}} to $\mathcal{I}'$ (with $M'=NM$). Note that for the Extended Recursive Greedy Algorithm in the framework, the approximate function $\bar{\Lambda}$ we use is calculated with respect to $\Lambda$ in the original instance $\mathcal{I}$. 
	
	To establish the performance guarantee of such procedure, we first analyze the quality of the greedy iterations computed by the Extended Greedy algorithm. At some iteration, let $\mathbf{f}$ be the single-path flow that corresponds to the path returned by the algorithm. Note that as Lemma \ref{lemma:gap} holds for all $P\in\mathcal{U}$, we have by Theorem \ref{thm:recgreedy2} that 
	\begin{align*}
	\Lambda(\mathbf{f},P)&\ge \frac{1}{(b+1)\cdot(\lceil\log d\rceil+1)}\Lambda(\mathbf{f^*},P)\\&\ge \frac{1}{(b+1)\cdot(\lceil\log d\rceil+1)}\Lambda(\mathbf{f'},P),
	\end{align*}
	where $\mathbf{f^*}$ is the flow returned by the exact greedy scheme with respect to $\Lambda$ and $\mathbf{f'}$ is the flow returned by that with respect to $'$. It follows that
	\begin{align}
	&\sum_{P\in\mathcal{U}}\Lambda'(\mathbf{f},P)=\sum_{P\in\mathcal{U}}\lfloor N_1\Lambda(\mathbf{f},P)\rfloor\\&\ge\sum_{P\in\mathcal{U}}\lfloor \frac{N_1}{(b+1)\cdot(\lceil\log d\rceil+1)}\Lambda(\mathbf{f}',P)\rfloor\nonumber\\
	&\ge\sum_{P\in\mathcal{U}}\left( \frac{N_1}{(b+1)\cdot(\lceil\log d\rceil+1)}\Lambda(\mathbf{f}',P)-1\right)\nonumber\\
	&\ge\sum_{P\in\mathcal{U}}\left( 1-\epsilon\right)\frac{N_1\Lambda(\mathbf{f}',P)}{(b+1)\cdot(\lceil\log d\rceil+1)}\nonumber\\
	&= \frac{1-\epsilon}{(b+1)\cdot(\lceil\log d\rceil+1)}\sum_{P\in\mathcal{U}}\Lambda'(\mathbf{f}',P).\nonumber
	\end{align}
	Therefore, the quality of the obtained approximate greedy solutions enjoys almost the same guarantee with respect to $'$. Then, invoking Theorem \ref{thm:robust}, denoting the optimal $N_0$-bounded strategy for $\mathcal{I'}$ as $w'$, we have that 
	\begin{align*}
	&\min_{P\in \mathcal{U}}\Lambda(w,P)\ge \frac{1}{N_1}\min_{P\in\mathcal{U}}\Lambda'(w,P)\\&\ge \left(\frac{N_0\left(1-\epsilon\right)}{(N_0+1)(b+1)\log M'\cdot(\lceil\log d\rceil+1)}\right)\min_{P\in\mathcal{U}}\frac{\Lambda'(w',P)}{N_1}\\
	&\ge \left(\frac{N_0\left(1-\epsilon \right)}{(N_0+1)(b+1)\log \frac{M}{\epsilon}\cdot(\lceil\log d\rceil+1)}\right)\min_{P\in\mathcal{U}}\frac{\Lambda'(w^*,P)}{N_1}\\
	&\ge \left(\frac{N_0\left(1-\epsilon \right)}{(N_0+1)(b+1)\log \frac{M}{\epsilon}\cdot(\lceil\log d\rceil+1)}\right)\min_{P\in\mathcal{U}}\frac{N_1\Lambda(w^*,P)-1}{N_1}\\
	&\ge \left(\frac{N_0\left(1-\epsilon \right)^2}{(N_0+1)(b+1)\log \frac{M}{\epsilon}\cdot(\lceil\log d\rceil+1)}\right)\min_{P\in\mathcal{U}}\Lambda(w^*,P).
	\end{align*}
\end{proof}
Note that the bound obtained in the proposition is essentially the same as that in Theorem \ref{thm:robust}.

\end{appendices}

\end{document}